\documentclass{llncs}

\include{include}

\begin{document}

\frontmatter          
\pagestyle{plain}  
\mainmatter              

\setcounter{secnumdepth}{3}
\setcounter{tocdepth}{3}

\title{Secure Two-Party Quantum Computation \texorpdfstring{\\}{} Over Classical Channels}

\author{Michele Ciampi\inst{1}, Alexandru Cojocaru\inst{2}, Elham Kashefi\inst{1,3}, Atul Mantri\inst{4}}

\institute{School of Informatics, The University of Edinburgh \\
\href{mailto:mciampi@ed.ac.uk}{michele.ciampi@ed.ac.uk}\\
\and
Inria \\ \href{mailto:dragos-alexandru.cojocaru@inria.fr}{dragos-alexandru.cojocaru@inria.fr}, \and
Laboratoire d'Informatique de Paris 6 (LIP6), Sorbonne Universit\'{e}, \\
\href{mailto:ekashefi@inf.ed.ac.uk}{ekashefi@inf.ed.ac.uk}
\and
Joint Center for Quantum Information and Computer Science (QuICS),\\ University of Maryland, College Park, USA \\ \href{mailto:amantri@umd.edu}{amantri@umd.edu}
}

\maketitle

\begin{abstract}
Secure two-party computation considers the problem of two parties computing a joint function of their private inputs without revealing anything beyond the output of the computation. 
In this work, we take the first steps towards understanding the setting where: 1) the two parties (Alice and Bob) can communicate only via a classical channel, 2) the input of Bob is quantum, and 3) the input of Alice is classical. Our first result indicates that in this setting it is in general \emph{impossible} to realize a two-party quantum functionality with black-box simulation in the case of malicious quantum adversaries. In particular, we show that the existence of a secure quantum computing protocol that relies only on classical channels would contradict the quantum no-cloning argument. 

We circumvent this impossibility following three different approaches. The first is by considering a weaker security notion called \emph{one-sided simulation} security. This notion protects the input of one party (the quantum Bob) in the standard simulation-based sense and protects the privacy of the other party's input (the classical Alice). We show how to realize a protocol that satisfies this notion relying on the learning with errors assumption. The second way to circumvent the impossibility result,  while at the same time providing standard simulation-based security also against a malicious Bob, is by assuming that the quantum input has an efficient classical representation. 

Finally, we focus our attention on the class of zero-knowledge functionalities and provide a compiler that takes as input a classical proof of quantum knowledge (PoQK) protocol for a \QMA{} relation $\mathcal{R}$ (a classical PoQK is a PoQK that can be verified by a classical verifier) and outputs a zero-knowledge PoQK for $\mathcal{R}$ that can be verified by classical parties. The direct implication of our result is that Mahadev's protocol for classical verification of quantum computations (FOCS'18) can be turned into a zero-knowledge proof of quantum knowledge with classical verifiers. To the best of our knowledge, we are the first to instantiate such a primitive. 

\end{abstract}

\section{Introduction}
Secure multi-party computation (\mpc{})~\cite{DBLP:conf/stoc/GoldreichMW87,yao1986generate} allows multiple distrusting parties to jointly compute any function of their joint input, such that no information is leaked about their private inputs apart from what can be inferred from the output of the computation~\cite{yao1982protocols,goldreich2009foundations}. Unfortunately, it is well known that \mpc{} is not possible with information-theoretic security when all but one party is corrupted.
Luckily, it is possible to circumvent this 
impossibility by considering computationally bounded adversaries or the existence of trusted third parties.

In the quantum world, secure quantum multiparty computation (\qmpc{}) and the case of quantum 2-party computation (\qtwopc{}) have been originally proposed in~\cite{van1997towards,Ben-Or2006,crepeau2002secure} and~\cite{dupuis2010secure,dupuis2012actively} and further explored in \cite{kashefi2017garbled,kashefi2017quantum,sun2018application,dulek2020secure}.  

All the previous works either realize post-quantum secure classical functionalities over classical networks or require all the involved parties to possess quantum resources and access to quantum channels. This means that quantum communication channels have been a must for securely implementing a quantum function and hence, put a rather heavy burden on the parties involved. Nonetheless, similar to the classical setting, we know that information-theoretic secure \qmpc{} with a dishonest majority is not possible. Therefore, the security achievable in \qmpc{} protocols with quantum channels are at best computational, despite requiring parties to have powerful quantum devices and access to the quantum channel. Hence, a natural question regarding the trade-off between the functionality achieved and the resources needed is the following:

\begin{center}
\emph{Do all parties require quantum devices and need to share quantum channels to securely evaluate a quantum function?} 
\end{center}

In this work, we study the case where two parties (Alice and Bob) have access to a classical channel, and only one party (Bob) has a quantum machine. Ideally Alice and Bob want to compute any arbitrary quantum computation on the classical-quantum (joint) input.
Getting insights into the minimum requirements for \qmpc{} protocols is important not only from the foundational perspective, but will also lay the stepping stone for practical \qmpc{} protocols. 

Unfortunately, our first result shows that achieving black-box simulation-based security in this setting is in general impossible. 
We recall that the notion of black-box simulation-based security
 guarantees the existence of a simulator that works by having only black-box access\footnote{Following~\cite{unruh2012quantum}, by black-box access here we mean that the simulator has oracle access to the unitary $M$ and $M^\dagger$, where $M$ is the quantum adversary.} to the adversary.
 To circumvent the impossibility, we follow three different (natural) approaches. The first is to weaken the security definition by considering the notion of \emph{one-sided simulation}. This notion protects the input of one party (the quantum Bob) in the standard simulation-based sense and protects only the privacy\footnote{By privacy of Alice's input we mean that the adversary is not able to figure out whether the input of Alice is $0$ or $1$.} of the other party's input (the classical Alice).
 The second approach we adopt to circumvent the impossibility is to restrict the class of functionalities that can be computed. In particular, we argue that if the quantum input of Bob has an efficient classical representation (known to Bob), then we can realize any functionality with standard black-box simulation-based security.
For our third approach, we restrict our attention to the  \emph{zero-knowledge} functionality for \QMA{}. In particular, we 
show how to realize the zero-knowledge functionality by properly combining classical proof of quantum knowledge (\PoQK) protocols\footnote{A classical proof of quantum knowledge (\CPoQK{}) is a proof of quantum knowledge (\PoQK) that can be verified by a classical verifier.} for \QMA{} with well-known existing post-quantum secure cryptographic primitives.
As an immediate corollary of our result, we obtain the first zero-knowledge proof of quantum knowledge protocol (\zkpoqk) with classical verifiers and quantum prover for a specific class of \QMA\ relations. The corollary is obtained by using as input of our compiler the \CPoQK{} protocol of Mahadev~\cite{mahadev2018classical,vidick2020classical}. To the best of our knowledge, our work is the first to provide such a primitive. We now provide more details on our results.

\subsection{Our Contributions}
All the results in this work concern secure quantum two-party computation (\qtwopc{}) over a classical channel\footnote{Unless otherwise specified, all the protocols and sub-protocols considered in this work allow the two parties to communicate only via classical channels. Moreover, only one party has access to a quantum machine}. We logically divide our contributions into three sets (we summarize our results in Figure~\ref{fig:diagram_res}).

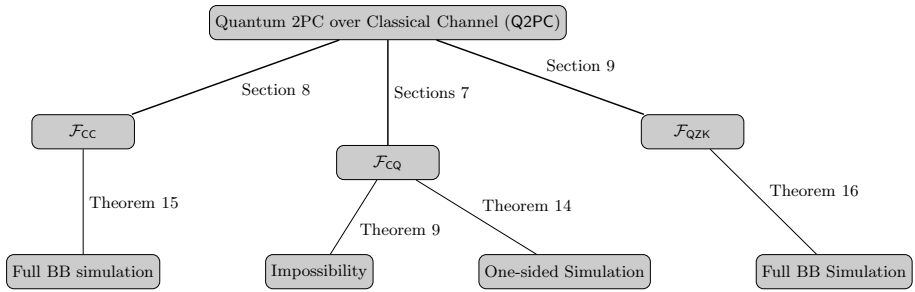
\begin{figure}[tbp]
    \centering 
    \tikzstyle{block} = [draw, fill=black!20, rectangle, rounded corners, minimum height=2em, minimum width=6em]
   \resizebox{\textwidth}{!}{
    \begin{tikzpicture}[auto, node distance=2cm]
        \node [block, left] (blk1) {Quantum 2PC over Classical Channel (\qtwopc{})};
        \node [block, below left=of blk1] (blk11) {\Fcc{}};
        \node [block, below =of blk1] (blk12) { \Fcq{}};
        \node [block, below right=of blk1] (blk13) {\Fqzk{}};
        \node [block, below =of blk11] (blk21) {Full BB simulation};
        \node [block, right=of blk21] (blk22) {Impossibility};
        \node [block, right=of blk22] (blk23) {One-sided Simulation};
        \node [block, right=of blk23] (blk24) {Full BB Simulation};
        \draw [draw,-,thick] (blk1) -- node {Section~\ref{sec:full_simulation_two_pc}} (blk11);
        \draw [draw,-,thick] (blk1) -- node {Sections~\ref{sec:oqfe}} (blk12);
        \draw [draw,-,thick] (blk1) -- node {Section~\ref{sec:compiler_zkpoqk}} (blk13);
        \draw [draw,-] (blk11) -- node {Theorem~\ref{thm:full_sim_qtwopc}} (blk21);
         \draw [draw,-] (blk12) -- node {Theorem~\ref{cor:imposs} } (blk22);
         \draw [draw,-] (blk12) -- node {Theorem~\ref{th:oqfeproof}} (blk23);
         \draw [draw,-] (blk13) -- node {Theorem~\ref{thm:zkpoqk}} (blk24);
    \end{tikzpicture}
    }
    \caption{Summary of our results. Our results on quantum two-party computation over a classical channel (\qtwopc{}) can be characterized into two levels: Functionality (first level) and Security (second level). We consider three different types of functionalities depending on the input of Alice and Bob. \Fcc{} is a joint quantum function with a classical input $x$ from Alice and a quantum input $\sigma_y$ from Bob, where the quantum input $\sigma_y$ has (efficient) classical description known to Bob. Similarly, \Fcq{} takes input $x$ and $\rho_y$, where $\rho_y$ is any quantum state with no classical description known to Bob. \Fqzk{} takes as input an instance $x$ that belongs to a \QMA{} language $L$ and a quantum state $\ket{\psi}$ from Bob, and returns to Alice either $0$ or $1$ depending on whether the verifier for the \QMA{} relation would have accepted $x$ and $|\psi \rangle$. The second level represents different notions of simulation-based security. The chart shows relevant sections and results that we study in this paper.}
    \label{fig:diagram_res}
\end{figure}

\paragraph{Security of \qtwopc{} with classical channel.}
What is the best possible security, achievable when a quantum two-party functionality is realized over classical networks? To answer this question, we analyze the difficulty underlying the process of achieving simulation-based security with black-box access to the adversary. To elaborate on this, we formalize a connection between the existence of a simulator for any malicious adversary in (quantum) two-party computation and the existence of an extractor (with the same run-time) in (quantum) \emph{agree-and-prove} (AaP) protocols~\cite{badertscher2019agree,vidick2020classical}. 
The notion of AaP is a generalization of the notion of proof-of-knowledge. In the latter, a prover and a verifier have a statement $x$ as a common input and the prover can construct a proof that convinces the verifier about the knowledge of a witness for $x$. In AaP the statements are not given as a common input, but it can be the output of an interactive protocol which acts as an agree-phase. At the end of the agree-phase, the prover and the verifier have a common input and might have some secret information. 
By observing that a $\sf{2PC}$ protocol can be used to realize the \emph{prove} phase of an AaP protocol and by leveraging on the impossibility of constructing  AaP protocols for quantum-money verification~\cite{vidick2020classical}, we also rule out the possibility of securely evaluating some functionalities in the setting where no quantum communication is allowed.

\begin{theorem}[Informal]
Secure quantum two-party computation with quantum input over a classical channel with black-box simulation is in general impossible.
\end{theorem}

Due to the above result, we present three relaxations ---one on the security model (we consider the notion of \emph{one-sided simulation}) and the other two on the class of functionalities that can be computed by the parties. 

\paragraph{One-sided simulation. }
We consider the notion of one-sided simulation which provides standard security against the classical party (Alice), and indistinguishably based security against the quantum party (Bob).
 For the sake of simplicity, we first present a protocol that realizes the functionality  \onetwooqfe{}. 
In \onetwooqfe{}, Alice has as input a bit $b$, while Bob's input is a single qubit state $\psi$. The target computation is 1 out of 2 possible functions $f_0$ and $f_1$. The \onetwooqfe{} functionality ensures that Alice obtains $f_b(\psi)$ without ``learning'' anything about the other function applied on Bob's input (i.e. $f_{1 \oplus b}(\psi)$), while Bob ``learns'' nothing about Alice's input $b$. 

 At a high level, our protocol is based on the measurement-based model of quantum computation. Alice's input is encoded in the measurement angles, Bob's input is the first layer of the underlying graph state and the last layer represents Alice's output. Alice remotely prepares the (auxiliary) qubits corresponding to the graph state using a cryptographic primitive known as remote state preparation (\rsp{})\footnote{In \rsp{}, a classical party (Alice) instructs a quantum party (Bob) to generate a quantum state remotely on Bob's side using classical communication only. The description of the generated quantum state is known to Alice but not to Bob. Such a task is only possible under computational assumptions.}. 
 We present two protocols based on this idea. One of which is non-interactive, secure against semi-honest Alice, and protects the input of Alice against a malicious Bob.  For this result, we rely on a  (computationally secure) classical-client remote state preparation protocol as a sub-module. Specifically, given the underlying connection to remote state preparation, this construction reduces the cryptographic complexity of \qtwopc{} to injective homomorphic trapdoor quantum one-way functions. 
 In the other protocol, we uplift the security of our semi-honest protocol against malicious Alice achieving one-sided simulation. Finally, we show how to extend the protocol to realize any functionality that belongs to \Fcq, where \Fcq{} denotes the class of two-input functionalities that admit one quantum input and one classical input\footnote{The subscript $\sf{CQ}$ for \Fcq{} denotes that the input of one party is classical whereas the input of the second party is quantum.}. Hence, we obtain the following result.

\begin{theorem}[Informal]
There exists a \qtwopc{} protocol that securely realizes \Fcq{}  with one-sided simulation, assuming the hardness of LWEs.
\end{theorem}

\begin{remark}
We emphasize that our approach is the first showing a relation between \rsp{}~\cite{CCKW18,cojocaru2019qfactory,gheorghiu2019computationally,badertscher2020security} and \oqfe, thus leveraging on the \rsp{} techniques to replace quantum channels with (more practical) classical channels in \mpc. We believe that this approach could naturally be generalized to the case where there are multiple classical and quantum parties.
\end{remark}

Our second result enhances the previous result by combining it with a sequence of zero-knowledge protocols, allowing us to obtain a fully simulatable secure protocol. To do this, we restrict ourselves to the class of quantum functionalities that admits only classical inputs (which we denote with \Fcc{}), then we show how to uplift the security of the one-sided-simulation protocol to make it fully simulatable.

\begin{theorem}[Informal]
There exists a \qtwopc{} protocol that realizes \Fcc{} with fully black-box simulation-based security, assuming the hardness of LWEs.
\end{theorem}

We also argue how to obtain one-sided simulation from circuit private Quantum Fully Homomorphic Encryption Approach (QFHE)~\cite{DBLP:journals/iacr/Malavolta20}. We can then use our compiler to obtain a fully simulatable \qtwopc{} protocol. We stress that QFHE without circuit privacy would not suffice to obtain our results.\footnote{The first circuit private QFHE protocol is proposed in~\cite{DBLP:journals/iacr/Malavolta20} which appeared after the first submission of this paper. We note that without the result proposed in ~\cite{DBLP:journals/iacr/Malavolta20} our protocol was the first to allow secure computation over a classical channel in the case of semi-honest parties.}

For our third result, instead, we use a completely different approach, showing how to promote any proof of quantum knowledge (\PoQK) protocol to a
\qtwopc{} protocol for the zero-knowledge functionality $\mathcal{F}_\mathsf{QZK}$ for \QMA. $\mathcal{F}_\mathsf{QZK}$ takes as input an instance and $x$ that belongs to a \QMA{} language $L$ and a quantum state $\ket{\psi}$ from Bob, and returns to Alice either $0$ or $1$ depending on whether the verifier for the \QMA{} relation would have accepted $x$ and $|\psi \rangle$. 
More precisely,  for this result, we use as the main tool a classical proof of quantum knowledge (\poqk{}) which enjoys the property of \emph{\mesind}. This property guarantees that the verifier can compute the protocol messages without looking at the messages received from the prover\footnote{Note that any public-coin protocol trivially enjoys this property.}.
We  propose a generic compiler that takes any \poqk{} that enjoys the property of \mesind\ and turns it into a secure \qtwopc{} protocol for the \Fqzk{} functionality.

 \begin{theorem}[Informal]
 If there exists a \poqk{} for the \QMA\ relation $\mathcal{R}$ that enjoys the property of  \mesind\ then there exists \qtwopc{} protocol that realizes \Fqzk{} (for the same relation $\mathcal{R}$) with fully black-box simulation-based security, assuming the hardness of LWEs.
 \end{theorem}

This result combined with the recent \poqk{} protocol proposed in~\cite{vidick2020classical} yields to a \qtwopc{} for \Fqzk{} for \QMA{} relations that satisfies some specific properties (we refer the reader to the technical part of the paper for more detail).
We stress that our compiler makes black-box use of the underling \poqk{} protocol, hence any advancement in the area of classical proof of quantum knowledge would immediately translate to a better \qtwopc{} for \Fqzk{}.

\subsection{Related Works}
The first work that studied the question of \mpc{} in the quantum domain is~\cite{crepeau2002secure}, where a secure $\sf{QMPC}$ protocol is proposed based on the ideas of verifiable quantum secret sharing of the inputs of the parties. This result has been extended in~\cite{crepeau2002secure,Ben-Or2006,lipinska2020secure}. In a series of works by Dupuis et al. \cite{dupuis2010secure,dupuis2012actively} the setting of two-party quantum computation is presented using tools from classical \mpc{} and quantum authentication codes developed in~\cite{aharonov2008interactive}. This protocol is generalized to the multi-party setting with a dishonest majority in a recent work by Dulek et al.~\cite{dulek2020secure} and extended to security with identifiable abort in~\cite{alonround20}. 
In a recent work of \cite{BY20}, the authors propose a garbling scheme for quantum circuits. 

A different approach inspired by delegated quantum computing~\cite{broadbent2009universal,dunjko2014composable} towards secure two-party computation, similar to a quantum analog of Yao's protocol from classical \mpc{}~\cite{yao1982protocols}, is studied in~\cite{kashefi2017garbled,kashefi2017quantum} and towards (composable) secure multi-party quantum computation in~\cite{kashefi2017multiparty,houshmand2018composable,mantri2019secure}. 

All previous works on secure two-party and multi-party quantum computation rely on quantum channels shared between parties and require more than one party to posses quantum devices. However, in our work, we propose a secure two-party quantum computation protocol over classical channels, removing the need for a quantum channel. Additionally, our constructions require only one party to have access to quantum resources.

 Unruh in~\cite{unruh2010universally}, building upon the works on quantum oblivious $\sf{OT}$~\cite{bennett1991practical} and \mpc{}~\cite{ishai2008founding}, proposed a UC-secure protocol for classical multi-party computations using only commitments and a quantum channel. More recently, secure (quantum) MPC based on quantum channels and one-way functions has been proposed  in~\cite{bartusek2020one,grilo2020oblivious}. Such a task is known to be impossible in a purely classical setting~\cite{gertner2000relationship}. It is not clear whether a quantum channel is necessary to achieve \mpc{} with quantum parties just relying on commitments. However, in this work, we show that a quantum channel is necessary to achieve general two-party quantum computation (\qtwopc{}) in the setting of malicious parties (and hence in the UC security model as well). Moreover, our modular construction of \qtwopc{} from \rsp{} establishes the latter primitive as a candidate for a universal primitive.  It is worth mentioning that due to this direct link, any further optimisation of the complexity of \rsp{} will also provide answers to the complexity of \qtwopc{}. This, in turn, will enhance our understanding about the resources required for important cryptographic primitives such as delegated quantum computing~\cite{fitzsimons2017private} and classical verification of quantum computing~\cite{gheorghiu2019computationally,mahadev2018classical}, in addition to $\sf{QMPC}$.

\section{Technical Overview}

\subsubsection*{Limitations of black-box secure \qtwopc{}.}

There are classical two-party functionalities that cannot be securely realized even if the parties have access to quantum resources~\cite{colbeck2007impossibility,salvail2009power,buhrman2012complete}. This automatically implies that quantum 2-PC cannot be achieved information-theoretically over classical channels. In this work, we show that there are quantum functionalities that cannot be securely realized (in a black-box way)  even against \emph{computationally-bounded} adversaries in the case where only classical channels are available. This features a striking trade-off between the resources needed to achieve the desired functionality and the level of security for quantum two-party computation (\qtwopc{}). We start by establishing a connection between black-box security for two-party computation and secure agree-and-prove protocols (which are a generalization of proofs/arguments of knowledge for both classical and quantum two-party functionalities). 

The high-level idea behind the connection of two-party computation and the existence of proof of knowledge (and agree-and-prove protocols) is the following. Saying that a two-party protocol realizes a functionality with black-box simulation-based security means that for every adversarial party there exists a simulator $\Sim$ that can extract the input from the malicious party. In the case of quantum adversaries (with quantum inputs), this implies that $\Sim$ must be able to extract Bob's input in quantum-polynomial time.
We can now consider a $\sf{2PC}$ protocol that realizes the \emph{prove} phase of the agree-and-prove protocol, where the prover (Bob) is proving the knowledge of a quantum secret.
In~\cite{vidick2020classical} the authors show that the existence of certain kinds of secure agree-and-prove protocols (in particular for quantum money scenario) implies cloning. This in turns implies that some functionalities cannot be securely realized in our model.

Next, we move towards our positive results, which are achieved by weakening either the security model or the class of functionality that we realize.

\subsection{Our Protocols}
The one-sided two-party classical-quantum setup consists of two parties --- Alice ($\alice$) and Bob ($\bob$) --- that have their private inputs and wish to perform a joint computation, but where only one of the parties receives the output. In the remainder of this work, we will use the following convention: a) Bob is a quantum party and b) Alice is a classical party, and is the only one receiving the output. One of the potential applications of such a setting is the following. Imagine a scenario where one of the parties, Bob, has a quantum database and the other party, Alice, queries the database in such a way that i) Bob would like to keep the entries of the database secure except the one which is queried, and ii) Alice would like to maintain the privacy of her requested query. In more detail, we consider the following setting.

\begin{enumerate}
    \item Alice has as input (a classical description of) a quantum function $f$, where $f$ has quantum input and classical output.
    \item Bob has as input a quantum state $\psi$.
    \item Alice obtains $f(\psi)$ and ``learns nothing'' more than this information. At the same time, Bob receives no output and ``learns nothing'' about $f$.
\end{enumerate}

In general, $\psi$ could be an arbitrary quantum state and since Alice is classical, $f$ denotes the quantum map that consists of a unitary $U$ followed by measurement in the computational basis. We provide a modular approach towards the construction of our protocol and prove its security in the one-sided simulation-based framework. At a high level, we first provide a protocol that achieves privacy against quantum Bob and (statistical) security against semi-honest Alice. Then, in the second construction, using cryptographic tools such as secure commitment schemes and zero-knowledge proof of knowledge, we uplift the security to full simulation-based security against malicious Alice.

To simplify the understanding of our protocol, we first present a construction for a simplified functionality called \onetwooqfe{} (Definition~\ref{def:one_sided_simulation_1_2_oqfe}), where Alice's input is a single bit $b$ and Bob's input is a single qubit state $\psi$. As mentioned before, this functionality ensures that Alice obtains $f_b(\psi)$ without ``revealing" anything about $f_{1 \oplus b}(\psi)$ to Alice as well as ``hiding" Alice's input $b$ from Bob. The notion of ``revealing" and ``hiding" is formalised using one-sided simulation framework (Section~\ref{sec:one-sided}). Aside from being instrumental in the construction and security proofs of the full \qtwopc{} protocol, this simple functionality of \onetwooqfe{} can be of independent interest.

The central idea behind the construction of this protocol (we refer to Protocol~\ref{protocol:quot_1_2_hbc_alice} for the formal description) is inspired by one-bit teleportation circuits where Bob's quantum input state is measured in one of two possible angles, which is dictated by Alice's (private) input. The output of such a (simple) two-party computation is obtained by Alice. Our protocol also relies on a remote state preparation (\rsp{}) --cryptographic primitive-- that enables an (honest) classical user to remotely prepare a quantum state on the (untrustworthy) quantum server, using only a classical communication channel.  \rsp{} plays an important role in our protocol to eliminate the need for quantum communication between Alice (user) and Bob (server). Although such a primitive cannot be information-theoretically secure, we use a computationally secure construction based on (post-quantum) cryptographic assumptions. While this means that the security of our protocol for \onetwooqfe{} holds only against semi-honest Alice, we can show that it holds in the statistical regime. On the other hand, for Bob, we show that the privacy of Alice's input is based on the hardness of the LWE problem. To sum up, we prove the following theorem. 

\begin{theorem}[Informal]
The exists a protocol that realizes the \onetwooqfe{} functionality which achieves privacy against malicious Bob and is statistically secure against semi-honest Alice.
\end{theorem}

To uplift the security from semi-honest Alice to a malicious Alice, we need to be able to validate the transcripts Alice is sending to Bob during the run of the protocol.
To do that, we let Alice and Bob engage in a coin-tossing protocol where only Alice receives the input. Then Alice proves that she has used the randomness generated from the coin-tossing to generate the messages of the \onetwooqfe{} protocol. We can therefore claim the following.

\begin{theorem}[Informal]
Assuming the hardness of LWE, there exists a protocol that realizes the \onetwooqfe{} functionality with one-sided simulation.
\end{theorem}

We then extend the previous construction to obtain a protocol that realizes any functionality \Fcq{} with one-sided simulation-based security.  More concretely, our construction for one-sided simulation secure \qtwopc{} is based on the measurement-based model of quantum computation, by combining the blind quantum computation protocol~\cite{broadbent2009universal} with \rsp{} as subroutines. 
Also, in this case, to enforce the honest behavior of Alice we use a combination of a coin-tossing protocol and zero-knowledge proofs. 

\begin{theorem}[Informal]
Assuming the hardness of LWE, there exists a protocol that realizes \Fcq{} with one-sided simulation.
\end{theorem}

It is easy to see that any protocol that realizes \Fcq{} can also be used to realize \Fcc{}.\footnote{We recall that this is the class of quantum functionality that accepts only inputs that have efficient classical descriptions.} Finally, we would like to remark that the cryptographic complexity of our \qtwopc{} proposal can be reduced to the cryptographic assumptions required to achieve classical-client \rsp. More specifically, our construction can be instantiated with injective homomorphic trapdoor quantum one-way functions and any zero-knowledge proof of knowledge.

\subsubsection*{From One-Sided Simulation to Full Simulation.} We propose a generic compiler that turns any protocol $\pi$ that realizes the class of functionalities \Fcc{} with one-sided simulation into a protocol  $\pi'$ that realizes the same class of functionalities with full simulation-based security. 
We recall that $\pi$ offers simulation-based security against the malicious Alice, and privacy against the malicious Bob. Hence, we just need to 
employ a mechanism that forces Bob to behave honestly and that at the same time allows extracting the input used by malicious Bob when running $\pi$\footnote{We recall that we need this extraction mechanism because we need to construct a simulator that in the ideal world acts on the behalf of Bob, hence it needs to be able to extract Bob's input.}.
One trivial solution would be to force Bob to provide, for each message, a classical zero-knowledge proof of quantum knowledge, that shows that Bob is executing correctly the protocol messages and that he knows what is the input and the randomness used in the computation. In the security proof against malicious Bob, we can then rely on the extractor of the classical proof of quantum knowledge protocol to retrieve the input of the malicious Bob and query the ideal functionality. Unfortunately, we are not aware of any such classical proof of quantum knowledge. Indeed all existing protocols (including ours) work only for a specific class of \QMA\ relations.

 Therefore, we follow a different approach. We require Bob to commit to its input $\phi_c$ and provide a classical proof of knowledge about the knowledge of the committed value (note that for this purpose a classical post-quantum secure proof of knowledge for NP relations suffices). After that the commitment and the proof have been computed, Alice and Bob run $\pi$, and when the last message of $\pi$ has been computed Bob provides a zero-knowledge proof that proves that the input used to compute the messages of $\pi$ is the same as the input committed in the very beginning of the protocol. Note that in this case we need a post-quantum secure zero-knowledge protocol for \QMA, but we require no properties of proof (argument) of quantum knowledge. We also observe that in the above protocol only Alice gets the output. However, this is without loss of generality as it is always possible to turn such a protocol into a protocol where also Bob gets the output (under the condition that the output is classical). Let us assume that Alice and Bob want to compute a function $f$ that belongs to the class \Fcc{} which takes two inputs $(\psi_c,x)$ and returns two outputs $y_B,y_A$. We now consider the function $\tilde f$ which takes two inputs $(k_1,k_2,\psi_c,x)$, where $k_1$ is the key of a one-time authentication scheme, and $k_2$ is the key for a one-time encryption scheme\footnote{We note that both this primitives can be instantiated information-theoretically.}. $\tilde f$ internally runs $f$, and outputs $(\mathsf{Enc_B},y_A)$, where $y_A$ represents the second output of $f$ and $\mathsf{Enc_B}$ represents the encryption computed using the key $k_2$ of the value $y_B$ authenticated with the key $k_1$.

Alice and Bob now can run the protocol $\pi'$ for $\tilde f$, and Alice, upon receiving $(\mathsf{Enc_B},y_A)$ sends $\mathsf{Enc_B}$ to Bob who decrypts it using $k_2$ and checks if the values are correctly authenticated with respect to $k_1$.  Note that this prevents Alice from seeing or tampering with 
the output dedicated to Bob.
Similar techniques have been used in many previous works~\cite{DBLP:conf/eurocrypt/AsharovJLTVW12,DBLP:journals/joc/LindellP09,DBLP:conf/crypto/IshaiKKP15}. This allows us to claim the next theorem.

\begin{theorem}[Informal]
Assuming the hardness of LWE, there exists a protocol that realizes \Fcc{} with simulation-based security.
\end{theorem}

\paragraph{\qtwopc{} from quantum fully-homomorphic encryption ($\sf{QFHE}$).}
We can obtain the above results starting from circuit private $\sf{QFHE}$. In particular, we can construct the following one-sided-simulation protocol that would make use of a classical $\sf{QFHE}$ scheme (for quantum computations). Alice and Bob first run a coin-tossing protocol. Then using the randomness resulting from this protocol, Alice generates the public key $pk$ of the $\sf{QFHE}$ and sends $pk$ to Bob together with proof that the public-key is generated accordingly the randomness obtained from the coin-tossing procedure. Bob then runs the function evaluation using $pk$, his input, and the function that needs to be computed and sends back the output to Alice.
Alice, upon receiving the encrypted message, decrypts it and obtains the outcome of the computation. We can argue that this protocol is a one-sided-simulation-based secure, Hence, we can use our generic compilers to achieve full simulation security.

\subsubsection*{\qtwopc{} for \Fqzk{}.}
For our last result, we do not put restrictions on whether the input of Bob can or cannot be represented classically, but we focus on the zero-knowledge functionality for \QMA{}, \Fqzk{}. Our approach in this case completely departs from what we have done so far. We first observe that to realize \Fqzk{} in our model, we \emph{only} need to construct a zero-knowledge proof of quantum knowledge for \QMA{} for classical verifiers.

We use as the main building block a classical proof of quantum knowledge $\Pi^{\mathsf{CPoQK}}$  (which admits a classical verifier). We recall that the messages of a \poqk{} protocol could leak information about the witness. Therefore, a first approach to solve this problem would be to let the prover and the verifier run $\Pi^{\mathsf{CPoQK}}$, where the messages of the prover are encrypted. At the end of the execution of $\Pi^{\mathsf{CPoQK}}$, the prover, using a zero-knowledge protocol $\Pi^\mathsf{ZK}$ proves an \NP{} statement of the following ``\emph{The verifier of $\Pi^{\mathsf{CPoQK}}$ would have accepted the transcript that consists of the encrypted messages}''. This approach unfortunately only yields a zero-knowledge proof for \QMA, since it is unclear how to argue that the overall protocol retains the \poqk{} property. To solve this issue, we add the following additional step at the beginning of the protocol. The prover commits to a secret key $\sk$ for an encryption scheme thus obtaining $\mathsf{com}$. Then he uses a zero-knowledge proof of knowledge protocol $\Pi^\mathsf{ZKPoK}$ to prove the following NP-statement ``\emph{I know the message committed in $\mathsf{com}$}''. The prover and the verifier then run $\Pi^{\mathsf{CPoQK}}$ as before (where the prover encrypts his messages), with the difference that we slightly modify the statement proved using $\Pi^\mathsf{ZK}$ as follows ``\emph{The verifier of $\Pi^{\mathsf{CPoQK}}$ would have accepted the transcript that consists of the encrypted messages, moreover the messages are encrypted with a secret key committed in $\mathsf{com}$}''.

We can prove that this protocol is zero-knowledge relying on the zero-knowledge property of $\Pi^\mathsf{ZKPoK}$ and $\Pi^\mathsf{ZK}$, on the hiding of the commitment scheme and the security of the encryption scheme. To prove that our protocol is a proof of quantum knowledge, we need to exhibit an extractor. Our extractor first extracts the secret key $\sk$ from the proof computed using $\Pi^\mathsf{ZKPoK}$, and then it can run the extractor of the protocol $\Pi^{\mathsf{CPoQK}}$ which exists by definition.

One limitation of our compiler is that it requires $\Pi^{\mathsf{CPoQK}}$ to have the property of \mesind. These properties require the verifier to be able to compute his messages without looking at the messages received from the prover. We note that this is a property enjoyed by all the public coin protocols and by existing protocols like the one proposed in~\cite{vidick2020classical}.

\subsection{Organization of paper}
In Section~\ref{sec:prelim}, we define the notations and relevant classical and quantum cryptographic primitives. In Section~\ref{sec:ideal_oqfe}, we present the definition of quantum two-party computation over a classical channel, ideal functionalities, and different simulation-based notions of security. In Section~\ref{app:imposs} we present the impossibility proof of black-box secure \qtwopc. Then, we present two concrete protocols for 1-out-of-2 oblivious quantum function evaluation and analyse the security against semi-honest Alice and malicious parties (one-sided simulation-based model) in Section~\ref{sec:1_2_oqfe_semihonest}. An extension from 1-out-of-2 oblivious quantum function evaluation (\onetwooqfe{}) to (general) secure two-party computation protocol along with its (one-sided simulation) security is presented in Section~\ref{sec:oqfe}.
In Section~\ref{sec:full_simulation_two_pc}, we show how to uplift our one-sided simulation secure \qtwopc{} protocol to a secure black-box \qtwopc{} assuming that Bob has a classical description of his input. Finally, in Section~\ref{sec:compiler_zkpoqk} we present a general compiler for constructing post-quantum zero-knowledge classical proof of quantum knowledge from simpler primitives.

In Appendix~\ref{app:ct} we describe more definitions related to complexity classes and cryptographic primitives concerning our work. Appendix~\ref{app_B_qfactory} describes the remote state preparation (\rsp{}) protocol used for our protocols, together with the security conditions it satisfies. In Appendix~\ref{app:onesided} we describe the \onetwooqfe{} protocol secure against malicious Alice.
Appendix~\ref{App_C_all_proofs} presents all the deferred proofs from Sections~\ref{sec:1_2_oqfe_semihonest}-\ref{sec:compiler_zkpoqk}.

\newcommand{\reg}[1]{{\mathsf{#1}}}
\section{Preliminaries \label{sec:prelim}}

\subsection{Notations}

In this paper when we talk about distributions being indistinguishable for any probabilistic polynomial-time (PPT) adversary we will use the symbol $\approx_c$, if they are indistinguishable for any quantum polynomial-time (QPT) adversary, we will use $\approx_q$ and if they are indistinguishable for an unbounded adversary we will use $\approx_u$. Additionally, when testing for equality we will use directly the symbol ``=''.
For a protocol $\mathcal{P} = (P_1, P_2)$ with two interacting algorithms $P_1$ and $P_2$ denoting the two participating parties, let $(r_1, r_2) \gets \left\langle P_1, P_2 \right\rangle$ denote the execution of the two algorithms, exchanging messages, with $P_1$'s output $r_1$ and $P_2$'s output $r_2$.
Let $\cA$ and $\cB$ be two Hilbert spaces. The set $\cL(\cA,\cB)$ is the set of all linear maps from $\cA$ to $\cB$. The set $\cL(\cA) = \cL (\cA,\cA)$ is the set of all linear maps on $\cA$ and the mapping $\varphi: \cL (\cA) \mapsto \cL (\cB)$ is also called super-operator. If $\varphi$ is completely positive and preserves the trace then such a super-operator is also known as quantum operation or CPTP map. We denote identity operator as $\mathbb{I}$ and $\cA \otimes \cB$ denotes the space of two such quantum registers. 
We will also denote by $M_Z$ a measurement of a quantum state in the computational basis. We will use the notation $R_x(-\alpha)$ to refer to the rotation of a single qubit around the $x$-axis with the angle $\alpha$, and $R_z(-\delta)$ to refer the rotation of a single qubit around z-axis with the angle $\delta$. For any function $f : A \rightarrow B$, we define the controlled-unitary $U_f$, as acting in the following way:
$U_f\ket{x}\ket{y} = \ket{x}\ket{y \oplus f(x)}$ for any $x \in A$ and $y \in B$,
where we name the first register $\ket{x}$ control and the second register $\ket{y}$ target. For more details on the quantum background, we refer the readers to~\cite{nielsen2002quantum}. 

In the rest of the section, we provide definition for black-box two-party quantum computation along with the definitions of classical and quantum primitives used in this work. Some parts of this section are taken from~\cite{unruh2012quantum,broadbent2019zero,coladangelo2020non}. The remaining definitions related to interactive quantum machines and (quantum) oracle are presented in Appendix~\ref{app:ct}.

\subsection{Classical and Quantum Cryptographic Primitives}
Polynomial time relation $\Rel$ is a subset of $\{0, 1\}^* \times \{0, 1\}^*$ such that membership of $(x, w)$ in $\Rel$ can be decided in time polynomial $|x|$. For a polynomial-time relation $\Rel$, we define the $\NP$ language $L_{\Rel} := \{x \, | \, \exists w \text{ such that } (x, w) \in Rel\}$. Next, we define proof of knowledge in both the classical and quantum settings. 

A \emph{Proof of Knowledge} ($\sf{PoK}$) is an interactive proof system for some relation $R$ such that if the verifier accepts a proof with respect to some input $x$ with high enough probability, then she is ``convinced'' that the prover ``knows'' some witness $w$ such that $(x,w) \in R$.  This notion is formalized by requiring the existence of an efficient \emph{extractor}~$\Extract$, that can return a witness for $x$ when given oracle access to the prover  
 (including ability to rewind its actions, in the classical case).

\begin{definition}[Post-Quantum Proof of Knowledge (from \cite{ananth2020concurrent})]
We say that an interactive proof system $(P, V)$ for a relation $R$ satisfies $(\epsilon, \delta)$-proof of knowledge property if the following holds: suppose there exists a malicious prover $P^*$ such that for every $x$ and quantum
state $\rho$ we have that:
$$ Pr[ (\tilde{\rho}, \text{d}) \leftarrow \langle P^*(x, \rho), V(x) \rangle \land \text{ d = accept}] = \epsilon $$
(where $\tilde{\rho}$ represents the output of $P^*$ and $d$ is the output of $V$), then there exists a quantum polynomial-time extractor $E$, such that:
$$Pr[(\rho', w) \leftarrow E(x, \rho)] = \delta $$
\end{definition}

\begin{definition}[Simulatability] We say a proof of knowledge system $(P, V)$ is simulatable if the following holds. Denote $\tilde{\rho}$ the state output by a malicious $P^*$ at the end of a protocol accepted by $V$ and $\rho'$ the output of the extractor at the end of a protocol accepted by $V$. Then $(P, V)$ is simulatable if $\rho \approx_q \rho'$.
\end{definition}

\begin{definition}[Black-Box Access to a Quantum Machine $M^*$, \cite{unruh2012quantum}] We say that a quantum algorithm $A$ has (rewinding) black-box access to a machine $M^*$ if by denoting with $U$ the unitary describing one activation of $M^*$, then $A$ can invoke $U$ (this corresponds to running $U$), or its inverse ${U}^\dagger$ (this corresponds to rewinding $M^*$ by one activation), and $A$ can also read/write a shared register $N$ used for exchanging messages with $M^*$ \end{definition}

\begin{definition}[Classical Proof of Quantum Knowledge ] $(P, V)$, where $P$ is $\qpt$ interactive machine and $V$ is a $\ppt$ interactive machine is a classical proof of quantum knowledge with knowledge error $k(\cdot)$ for \QMA{} relation $Rel$ if the following are satisfied:
\begin{itemize}
    \item \textit{Correctness}: For any $(x, \rho) \in Rel$, we have: $Pr[\langle P(\rho), V \rangle (x) = 1] = 1 - negl(|x|)$
    \item \textit{Soundness}: There exists QPT algorithm $E$, called extractor such that: for any prover $P^*$ and any $x$ accepted by $V$ when interacting with $P^*$ with probability $\epsilon(x) > k(x)$, we have:
    $$ Pr[(x, \rho) \in Rel : \rho \leftarrow E^{P^*}(x)) ] \geq 1 - \delta(\epsilon)$$
where $\delta$ is such that $\delta(\epsilon) < 1$ for all $\epsilon > k$.
\end{itemize}
\label{def:cpoqk}
\end{definition}

\subsection{Useful Sub-Protocols}

 \begin{definition}[Remote State Preparation] A resource $\cS$ is a Remote State Preparation (\rsp) resource if it outputs on the right interface (Bob's interface) a quantum state $\rho$ and the left interface (Alice's interface) a classical description of a state $\rho'$ such that the states $\rho$ and $\rho'$ are close in trace distance.
  
The security guarantee is that Bob should not be able to learn anything more about the classical description of $\rho'$ (known to Alice) than what he can learn from a single copy of the state $\rho$.
\end{definition}

For a more general definition of remote state preparation, we refer to \cite{badertscher2020security}.
\textbf{Instantiation.} A remote state preparation primitive can be instantiated using the QFactory protocol \cite{cojocaru2019qfactory}. This protocol is described in Appendix~\ref{app_B_qfactory}. 

In our constructions we employ QFactory as an \rsp{} primitive, which ensures that Bob produces one of the BB84 states 
$\left\{H^{B_1}X^{B_2}\ket{0} \, | \, B_1, B_2 \in \{0, 1\}\right\}$ on his side. At the same time, Alice learns the complete classical description of this state, consisting of the 2 bits - $B_1$ and $B_2$. Roughly speaking, the security of QFactory guarantees that Bob learns nothing about $B_1$.

\begin{definition}[Delegated Quantum Computation (Informal)] A resource $\cS$ is a Delegated Quantum Computing resource if it takes on the left interface (Alice's interface) a description of a computation $\Psi$ such that $\Psi$ encodes both the input state $\psi$ and the unitary $U$. It outputs on the left interface (Alice's interface) the output $U(\psi)$  and on the right interface outputs the circuit dimensions. 
\end{definition}

\textbf{Instantiation.} A delegated quantum computing over quantum channel and classical channel can be instantiated using the protocol \cite{broadbent2009universal} and \cite{badertscher2020security}, respectively.  

\paragraph{Measurement-Based Quantum Computing model ($\sf{MBQC}$) and Universal Blind Quantum Computing protocol (\ubqc{}).} 
Some of our protocols rely upon the principles of the measurement-based model for quantum computing. This model is known to be equivalent to the quantum circuit model (up to polynomial overhead in resources).
For more details on measurement-based quantum computation, we refer the reader to~\cite{raussendorf2001one,nielsen2006cluster} and to this excellent tutorial~\cite{browne2006one}. This model of quantum computing is particularly useful in a well-known delegated quantum computing protocol, also known as the universal blind quantum computing (\ubqc{}) protocol. This was first proposed in \cite{broadbent2009universal}, where a computationally-weak user (Alice) delegates an arbitrary quantum computation to a quantum server (Bob), in such a way that her input, the quantum computation, and the output of the computation are information-theoretically hidden from Bob. In \ubqc{}, Alice prepares single qubits of the form $\Ket{+_{\theta}} := \frac{1}{\sqrt{2}}(\ket{0} + e^{i \theta}\ket{1})$, where $\theta \in \{0, \frac{\pi}{4}, \ldots, \frac{7\pi}{4}\}$ and send these quantum states to Bob at the beginning of the protocol, the rest of the communication between the two parties being classical.

\section{Our Model and Impossibility of Black-Box \texorpdfstring{\qtwopc}{}} \label{sec:ideal_oqfe}

\begin{definition}[\qtwopc{} over classical channel]
\label{def:q2pc}
  An $m$-round (classical Alice, quantum Bob) quantum protocol $\cP = (\alice, \bob, m)$ over classical communication channel consists of: 
\begin{enumerate}
    \item input spaces $S_{0}$ and $S'_{0}$ consisting of (classical) input $x_A$ and (quantum) input $\psi_B$ for parties $\alice$ and $\bob$, respectively.
    \item memory spaces $\textsf{S} \coloneqq (S_{1}, \ldots, S_{m})$ for $\alice$ and  $\textsf{S}' \coloneqq (S'_{1}, \ldots, S'_{m})$ for $\bob$ and (classical) communication spaces $\textsf{N} \coloneqq (N_{1}, \ldots, N_{m})$ and $\textsf{N}' \coloneqq (N'_{1}, \ldots, N'_{m})$. 
    \item an $m$-tuple of stochastic operations $(\opA_1, \ldots, \opA_m)$ for $\alice$, where $\opA_1: \cL(S_0) \mapsto  \cL(S_1 \otimes N_1)$, and $\opA_i: \cL(S_{i-1} \otimes N'_{i-1} ) \mapsto  \cL(S_i \otimes N_i)$, ($2 \leq i \leq m$).
    \item  an $m$-tuple of quantum operations $(\opB_1, \ldots, \opB_m)$ for $\bob$, where $\opB_i: \cL(S'_{i-1} \otimes N_{i} ) \mapsto  \cL(S'_i \otimes N'_i)$, ($1 \leq i \leq m-1$) and $\opB_m: \cL(S'_{m-1} \otimes N_m) \mapsto  \cL(S'_m)$.
\end{enumerate}
\end{definition}


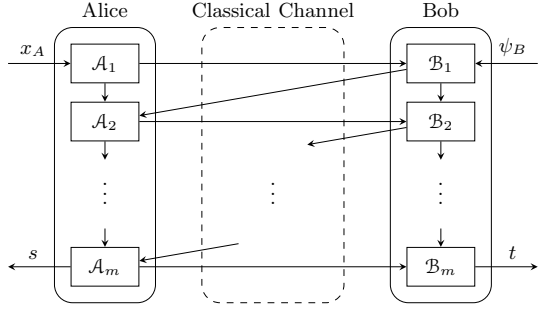
\begin{figure}[htb]
 \begin{center}
 \resizebox{.6\textwidth}{!}{
 \begin{tikzpicture}[opnode/.style={minimum width=1.05cm,minimum height=.6cm}]
 \small
 \def\t{4.1}
 \def\u{2.6}
 \def\v{.9}
 \node[draw,opnode] (a1) at (-\u,0) {$\opA_1$};
 \node[draw,opnode] (a2) at (-\u,-\v) {$\opA_2$};
 \node[opnode] (a3) at (-\u,-2*\v) {};
 \node[opnode] (a35) at (-\u,-2.12*\v) {$\vdots$};
 \node[opnode] (a4) at (-\u,-2.5*\v) {};
 \node[draw,opnode] (a5) at (-\u,-3.5*\v) {$\opA_m$};
 \node[draw,inner sep=.25cm,fit=(a1)(a5),rounded corners=8] (a) {};
 \node[above=1pt] at (a.north) {Alice};
 \node (alice1) at (-\t,0) {};
 \node (alice2) at (-\t,-3.5*\v) {};
 \node[draw,opnode] (b1) at (\u,0) {$\opB_1$};
 \node[draw,opnode] (b2) at (\u,-\v) {$\opB_2$};
 \node[opnode] (b3) at (\u,-2*\v) {};
 \node[opnode] (b35) at (\u,-2.12*\v) {$\vdots$};
 \node[opnode] (b4) at (\u,-2.5*\v) {};
 \node[draw,opnode] (b5) at (\u,-3.5*\v) {$\opB_{m}$};
 \node[draw,inner sep=.25cm,fit=(b1)(b5), rounded corners=8] (b) {};
 \node[above=1pt] at (b.north) {Bob};
 \node (bob1) at (\t,0) {};
 \node (bob2) at (\t,-3.5*\v) {};
 \node[opnode] (c1) at (0,0) {};
 \node[opnode] (c2) at (0,-1.5*\v) {};
 \node[opnode] (c35) at (0,-2.12*\v) {$\vdots$};
 \node[opnode] (c4) at (0,-3*\v) {};
 \node[opnode] (c5) at (0,-3.5*\v) {};
 \node[draw,minimum width=2.2cm,inner sep=.25cm,fit=(c1)(c5), rounded corners=8, dashed] (c) {};
 \node[above = 1pt] at (c.north) {Classical Channel};
 \draw[sArrow] (alice1.center) to node[auto,pos=.4] {$x_A$} (a1);
 \draw[sArrow] (bob1.center) to node[auto,pos=.4,swap] {$\psi_B$} (b1);
 \draw[sArrow] (a1) to (b1);
 \draw[sArrow] (a1) to (a2);
 \draw[sArrow] (b1) to (a2);
 \draw[sArrow] (b1) to (b2);
 \draw[sArrow] (a2) to (b2);
 \draw[sArrow] (a2) to (a3);
 \draw[sArrow] (b2) to (c2);
 \draw[sArrow] (b2) to (b3);
 \draw[sArrow] (c4) to (a5);
 \draw[sArrow] (a4) to (a5);
 \draw[sArrow] (b4) to (b5);
 \draw[sArrow] (a5) to (b5);
 \draw[sArrow] (a5) to node[auto,swap,pos=.6] {$s$} (alice2.center);
  \draw[sArrow] (b5) to node[auto,pos=.6] {$t$} (bob2.center);
 \end{tikzpicture}}
 \end{center}
 \caption[Two-party quantum protocol with classical Alice and quantum Bob~\cite{dunjko2014composable}]{\label{fig:q2pc}Two-party quantum protocol with classical Alice and quantum Bob.  The left and the right box represents the parties (classical) Alice and (quantum) Bob, and the dashed box depicts the classical communication channel between them. The input of Alice is classical while Bob's input could be quantum or classical. Both the parties obtain classical output. }
 \end{figure}

Let $\cF$ be a joint quantum function $\cF: \cL(\cA_{in}, \cB_{in}) \mapsto \cL(\cA_{out})$ that (classical) Alice and (quantum) Bob would like jointly compute on their private input. Without loss of generality, we assume that only Alice obtains the (classical) output. We define two sets of quantum functionalities. Let  $\cF_{CC}$ be a joint quantum function with input $x$ and $y$ from Alice and Bob, respectively, where $x$ and $y$ are the (efficient) classical description of their quantum input. Similarly, $\cF_{CQ}$ be a joint quantum function with input $x$ and $\rho_y$ from Alice and Bob, respectively. In this case $x$ is an (efficient) classical description of Alice's (quantum) input and
$\rho_y$ is an arbitrary quantum state representing Bob's input. We also define the zero-knowledge functionality $\mathcal{F}_\mathsf{QZK}$ (that is parametrized by a \QMA{} relation). $\mathcal{F}_\mathsf{QZK}$ takes as input an instance $x$ that belongs to a \QMA{} language $L$ and a quantum state $\ket{\psi}$ from Bob, and returns to Alice either $0$ or $1$ depending on whether the verifier for the \QMA{} relation would have accepted $x$ and $|\psi \rangle$.

\subsection{Simulation Based Security}
\label{sec:one-sided}
Let $\myreal_{\Pi, \mathcal{A}(z), i}(x, y, 1^{\lambda})$ and $\ideal_{\mathcal{F}, \mathcal{S}(z), i}(x, y, 1^{\lambda})$ be the output in the real and ideal execution for $\cF$ when the adversary $\mathcal{A}$ is controlling party $i \in\{\alice,\bob\}$ with the auxiliary input $z$. The $view_{\Pi, \mathcal{A}(z), i}(x, y, 1^{\lambda})$ denote the view of the adversary $\mathcal{A}$ after a real execution of $\Pi$. 

\begin{definition}[One-Sided Secure Realization of $\mathcal{F}$]\label{def:one}
 We say a protocol $\Pi = (\alice, \bob)$, where $\alice$ is classical and $\bob$ is quantum, securely computes $\mathcal{F}$ with one-sided simulation (adapted from Def.2.6.2 in \cite{CLBook10})) if for all inputs $(x,y) \in D(A_{in} \otimes B_{in}\otimes R)$ we have:
     \begin{enumerate}
     \item For every non-uniform QPT adversary $\mathcal{A}$ controlling $\alice$ in the real model, there exists a non-uniform QPT adversary $\mathcal{S}$ for the ideal model, such that:    
    \begin{equation} \label{eq:gen_alice_sec_ours}
         \{\myreal_{\Pi, \mathcal{A}(z), \alice}(x, y, 1^{\lambda})\}_{x, y, z, \lambda} \approx_{c} \{\ideal_{\mathcal{F}, \mathcal{S}(z), \alice}(x, y, 1^{\lambda})\}_{x, y, z, \lambda}
    \end{equation}
     \item For every non-uniform QPT adversary $\mathcal{A}$ controlling $\bob$, we have that the following distributions are indistinguishable to any QPT distinguisher $\mathcal{D}$:
     \begin{equation} \label{eq:gen_bob_sec_ours}
         \begin{split}
         & \{view_{\Pi, \mathcal{A}(z), \bob}(x, y, 1^{\lambda})\}_{x,x',y, z, \lambda} \approx_q \{ view_{\Pi, \mathcal{A}(z), \bob}(x',y, 1^{\lambda})\}_{x,x',y, z, \lambda} \\
         & \text{where } |x| = |x'|. 
         \end{split}
     \end{equation}
 \end{enumerate}
 \end{definition}

\sloppy Similarly, we extend the definition of black-box 2-party computation from~\cite{C:OstRicSca15} to quantum functionalities.

\begin{definition}[Black-box (Quantum) 2-party computation]\label{def:Qtwopc}
We say that a protocol $\Pi = (\alice, \bob)$, where $\alice$ is classical and $\bob$ is quantum, securely computes $\cF$ if for every $i\in\zo$, for all inputs $(x,y) \in D(A_{in} \otimes B_{in}\otimes R)$, for every non-uniform quantum-polynomial-time ($\qpt$) adversary $P_i^\star$ controlling $P_i$ in the real model, there exists a non-uniform  quantum polynomial-time adversary $\Sim_i$ (having black-box access to $P_i^*$) for the ideal world such that:
\[
\{\myreal_{\Pi,P_i^\star(z)}(x, y, 1^\lambda)\}_{x, y, z, \lambda} \approx_q \{\ideal_{\cF,\Sim_i(z)}(x, y, 1^\lambda)\}_{x, y, z, \lambda}
\]
where $z$ is the auxiliary input.
\end{definition}

\section{Impossibility of Secure Quantum 2PC }
\label{app:imposs}

In this section, we show that it is in general impossible to obtain a protocol that is fully simulatable (i.e. that satisfies  Definition~\ref{def:Qtwopc}) when only classical channels are available. To show this we will argue that  for all the protocols realizing a specific function there exists an adversarial Bob that cannot be simulated.
We note that the same no-go argument remains valid for \oqfe{} as well.

\begin{theorem}
\label{cor:imposs}
Secure quantum two-party computation over a classical channel with fully black-box simulation is not possible.
\end{theorem}

\begin{proof}
We choose the same notations for Alice and Bob as presented in Definition~\ref{def:q2pc} and assume that Bob has a quantum input while Alice's input is completely classical. For simplicity, we assume that only Alice obtains the output. Let us define a cheating strategy in the following way. We take an unclonable state, $\rho$ in the input space $S_0'$ of Bob, while the general quantum operations ($B_1, \ldots, B_m$) performed on Bob's end are such that the interaction between Alice and Bob is perfectly \emph{non-destructive} (as defined in Def.~\ref{def:nondestructive-interaction}). Similarly, Alice's input is completely classical and denoted as $a$. 

In the case when the functionality takes the (classical) input, we define $a$ as a triple $(x,public, secret)$ and similarly, Bob's input $\rho_w$ can be divided into classical $(x,public)$ as well as quantum state $\rho$. We define the two-party functionality such that it takes the input $(a,\rho_w)$, runs a ($\qpt$) verification algorithm $\reg{Ver} (a,\rho_w) \equiv \reg{Ver}(x, public, secret, \rho)$ and outputs a bit $b\in\{0,1\}$ to Alice. The idea behind such functionality is to emulate the proof phase of a specific secure agree-and-prove scheme (Definition~\ref{def:AaPqmoney}), that are generalized proofs of (quantum) knowledge for the quantum money functionality (as defined in~\cite{vidick2020classical}). In other words, non-destructive secure agree-and-prove protocols between a classical verifier and the quantum prover (for a quantum money scenario) are a special case of quantum two-party computation over a classical channel. However, non-destructive proofs of quantum knowledge for such a quantum money scenario imply cloning (Theorem~\ref{thm:no-nondestructive}).  
\end{proof}

As mentioned before, a natural step forward is to relax the security requirement to one-sided simulation. We present a concrete \oqfe{} protocol and analyse the security in the case when either of the party is malicious. In particular, we show simulation-based security against a $\qpt$ Alice and input privacy against a $\qpt$ Bob.  Such compromise seems inevitable given our previous result and the challenges to construct a quantum proof of knowledge, which shows that full simulation-based security in the presence of malicious adversaries is not possible.

\section{\texorpdfstring{\onetwooqfe}{}} \label{sec:1_2_oqfe_semihonest}


Before describing a candidate construction for \oqfesmall{} that satisfies the one-sided simulation-based security (Def.~\ref{def:one}), we introduce a simplified functionality \onetwooqfe{}, which can be understood as Alice having 2 possible functions $f_0$ and $f_1$ that she can choose from. The reason for introducing \onetwooqfe{} is that it sets the stage for the general two-party quantum computing functionality, is simpler in construction, and provides a good intuition for the \oqfe{} protocol.

We model \onetwooqfe{} functionality, denoted as \onetwooqfefunct{}, in the following way: Bob's private quantum computation is parameterized by a quantum state $\ket{\psi_{in}}$ and a set of angles $(\phi_0, \phi_1)$. The measurement outputs $s_i$ are then given by the computational basis measurement of the unitary $R_x \left(\phi_{i} \right)$ on the input state $\ket{\psi_{in}}$ for $i \in \{0,1\}$. For simplicity, we choose $\phi_0 = 0$ and $\phi_1 = \pi/2$. Therefore, the functionality \onetwooqfefunct{} is given by:
\begin{equation}
((s_0, s_1), b) \rightarrow (\lambda, s_b, \epsilon) 
\end{equation}
where $\lambda$ denotes the empty string, (1- $\epsilon$) is the probability of Alice obtaining the correct outcome and $s_b = M_Z R_x \left( - \frac{\pi}{2} \cdot b \right) \ket{\psi_{in}}$

\begin{definition} \label{def:one_sided_simulation_1_2_oqfe}
(Ideal Functionality \texorpdfstring{\onetwooqfefunct}{Fqot})
A \onetwooqfe{} functionality \onetwooqfefunct{} is defined as follows. When both the parties ($\alice$ and $\bob$) are honest then \onetwooqfe{} takes the input $b$ and $(s_0, s_1)$ from Alice and Bob, respectively and outputs $s_b$ at Alice's side. Here $(s_0, s_1)$ corresponds to measurement output of Bob's (server's) private quantum computation and $b$ denotes the choice of computation Alice (user) wishes to retrieve.
\end{definition}

In this section, we present a two-party quantum computing protocol over a classical channel inspired by the one-bit teleportation circuit~\cite{zhou2000methodology,gottesman1999demonstrating}. We also employ a cryptographic primitive known as remote state preparation (\rsp) as a subroutine between a classical party (Alice) and quantum party (Bob) to delegate the (quantum) computation from Alice to Bob. \rsp{} was first proposed in \cite{cojocaru2019qfactory,gheorghiu2019computationally} in the context of classical-client delegated quantum computing. This simple two-party protocol, which we call as \onetwooqfe{}, serves as a stepping stone for the general two-party quantum computing protocol presented in Section~\ref{sec:full_simulation_two_pc}. 

In \onetwooqfe{}, we have two parties: a (classical) Alice and a (quantum) Bob, where Alice's input is represented by a bit $b$ and Bob has a quantum input $\Ket{\psi_{in}}$. As mentioned before, we will consider Alice to be the party that receives the output at the end of the protocol. Without loss of generality, we assume that the goal of Alice and Bob is to perform a unitary operation $U_b$ from a set $\cU$ on Bob's private input state $\Ket{\psi_{in}}$. Here, the set of unitaries $\cU := \{U_b\}_{b \in \{0, 1\}}$ is known to both parties and we denote the output of the joint computation as  $\Ket{\psi_{out}}$, where $\Ket{\psi_{out}} := U_b\Ket{\psi_{in}}$. To this end, we present two \onetwooqfe{} protocols, one with the semi-honest Alice (Protocol~\ref{protocol:quot_1_2_hbc_alice}) and the other with the malicious Alice (Protocol~\ref{protocol:quot_1_2_malicious_alice}). In both these protocols, Bob could be completely malicious.  In this section, we will denote Alice's input with a bit and Bob's input with a single qubit quantum state. Furthermore, Alice's input bit $b$ is encoded using an angle $\phi_b$, where $b = 0$ corresponds to angle $\phi_0 := 0$ and $b = 1$ corresponds to $\phi_1 := \pi/2$ while Bob's private input is represented as $\ket{\psi_{in}}$.

\subsection{Semi-honest Alice}
Our first construction is a 2-message, non-interactive protocol for semi-honest Alice (Protocol~\ref{protocol:quot_1_2_hbc_alice}) and proceeds in two-stages.
In stage I (step 1 and step 2) of \onetwooqfeprotone{}, Alice encrypts her private bit $b$ (parameterized with angle $\phi_b$) using one-time pad with the random key $\theta_2$ and $r_A$, where both $\theta_2$  and $r_A$ are randomly chosen bits, to obtain an angle $\delta$. This is given by the following equation: 
\begin{equation}\label{eq:delta}
\delta := \phi_b + \theta_2 \cdot \frac{\pi}{2} + r_A \cdot \pi
\end{equation}
Since $\phi_b = b\cdot\pi/2$, we can rewrite the angle $\delta$ as
\begin{equation}\label{eq:delta2}
\delta = (b + \theta_2) \cdot \frac{\pi}{2} + r_A \cdot \pi = \frac{\pi}{2} \cdot \left[b + \theta_2 + 2r_A \bmod 4 \right] 
\end{equation}
From here on we will drop the $\pi/2$ coefficient and refer to the angle $\delta$ as  $b + \theta_2 + 2r_A \bmod 4$ instead of $\frac{\pi}{2} \cdot \left[b + \theta_2 + 2r_A \bmod 4 \right]$. 
As the next step, Alice and Bob execute the remote state preparation ($\sf{RSP}$) subroutine (Protocol~\ref{protocol:rsp}) to allow Alice to remotely prepare the following single-qubit state on Bob's side while hiding $\theta$ from him. The remotely prepared quantum state lies in the (X,Y)-plane of the Bloch sphere and is parameterised with an angle $\theta$. This is denoted in the following way.
\begin{equation}
\label{eq:rspstate}
\Ket{\psi_A} = \Ket{+_{\theta}} \coloneqq \frac{1}{\sqrt{2}}(\ket{0} + e^{i \theta}\ket{1})
\end{equation} 
where $\theta \in \{0,\pi/2,\pi,3\pi/2\}$ is an angle represented as a two bit string $(\theta_1, \theta_2) \in \{0, 1\}^2$ and $\theta_2$ is exactly the bit used above by Alice in angle $\delta$ (Eq.~\ref{eq:delta2}).
It is important to emphasize that the exact construction of remote state preparation used here relies on a two-way communication classical channel. Finally, Alice transmits to Bob the classical message $\delta$ along with the classical messages required in Protocol~\ref{protocol:rsp}. 
In stage II (step 3-6), Bob initializes an ancillary register, $\Ket{\psi_B}$, in the $|+\rangle := \frac{1}{\sqrt{2}}(|0\rangle + |1\rangle)$ state and performs entangling operations, controlled-Z ($\sf{CZ}$) gates, as shown in Fig.~\ref{fig:q_comp_bob_q_alice}.

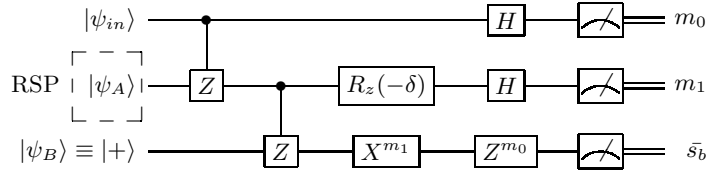
\begin{figure}[ht]
\[
\Qcircuit @C=1.7em @R=1.2em{
	& \lstick{\Ket{\psi_{in}}} &  \ctrl{1}    & \qw & \qw & \gate{H}  & \meter & \cw \, \, \,  \, \, \,  \, \, \,  \, \, \, \, m_0
    \\
    & \lstick{\text{RSP} \quad \Ket{\psi_A}} & \gate{Z}  & \ctrl{1} &  \gate{R_z(-\delta)} & \gate{H} & \meter & \cw \, \, \,  \, \, \,  \, \, \,  \, \, \, \, m_1
 \\
    & \lstick{\Ket{\psi_B} \equiv \Ket{+} } &  \qw  & \gate{Z}  & \gate{X^{m_1}} & \gate{Z^{m_0}} & \meter & \cw \, \, \,  \, \, \,  \, \, \, \, \, \, \, \, \, \,  \bar{s_b} \gategroup{2}{1}{2}{1}{2.8em}{--}
}
\]
\caption{Quantum computations performed by Bob in steps 3-7 of Protocol \ref{protocol:quot_1_2_hbc_alice}. The quantum state $|\psi_{in}\rangle$ represents Bob's input. The quantum state $|\psi_{A}\rangle$ (inside the dashed box) is generated using classical-client remote state preparation Protocol~\ref{protocol:rsp} between Alice and Bob. The angle $\delta$ encodes Alice's private (classical) input $b$.  }
\label{fig:q_comp_bob_q_alice}
\end{figure}

In the end, Bob measures the first two registers corresponding to his private input state and to the quantum state obtained from the $\sf{RSP}$ procedure. The two measurements are performed in the Hadamard and $\{\ket{\pm_{\delta}}\}$ basis, respectively, where $\ket{\pm_{\delta}} \coloneqq R_z(-\delta)|\pm\rangle$ and $R_z(-\delta)$ is the rotation around z-axis with the angle $\delta$. This step results in the measurement outcomes $m_0$ and $m_1$ on Bob's side. Finally, Bob performs a correction operator $Z^{m_0} X^{m_1}$ on the ancillary register and measures it in computational basis to obtain $\bar{s_b}$. Bob sends $m_0$ and $\bar{s_b}$ to Alice. Finally, Alice can compute her output $s_b$ as follows: $s_b \coloneqq \bar{s_b} \, \oplus \, \theta_1 \, \oplus \, r_A \, \oplus \, m_0 \cdot b$. \\

\begin{breakablealgorithm}
\caption{\texorpdfstring{\onetwooqfe}{} Protocol, \onetwooqfeprotone{}, with Semi-Honest Classical Alice} \label{protocol:quot_1_2_hbc_alice}
\vspace{\baselineskip}
\noindent \textbf{Inputs:} Bob: single qubit state $\Ket{\psi_{in}}$ and Alice: $b \in \{0,1\}$ \\
\noindent \textbf{Output:} Alice: $s_b$, where $s_b \coloneqq M_Z [R_x\left(-\frac{\pi}{2} \cdot b \right)\ket{\psi_{in}}]$ \\ 
\textbf{Requirements:} 
A 2-regular trapdoor one-way family $\mathcal{F}$ and homomorphic hardcore predicate $h_k$.
\begin{enumerate}
    \item Alice and Bob run a classical-client remote state preparation protocol (See, Appendix~\ref{protocol:rsp}). Alice obtains $\theta_2$ and Bob obtains $\Ket{\psi_A}$ (See, Eq.~\ref{eq:rspstate}) as follows:
    \begin{enumerate}
        \item Alice runs the algorithm $(k,t_k) \leftarrow \text{Gen}_{\mathcal{F}}(1^n)$ and sends $k$ to Bob. 
        \item Bob prepares a quantum state: $\ket{0}_{R1}\otimes \ket{+}^{n}_{R2}\otimes \ket{0}^{m}_{R3}$, applies $U_{f_k}$ using the second register (R2) as control and the third (R3) as target and measures the R3 register in the computational basis to obtain the (classical) outcome $y$.
        \item Bob applies $U_{h_k}$  on R2 as control and R1 as target, and measures the R2 register to obtain measurement outcome $m$. Bob finally applies $HR_z(-\pi/2)$ on R1 register to obtain the desired state $|\psi_A\rangle$. Bob sends $m_{qf} := (y,m)$ to Alice. 
    \end{enumerate}
    \item Alice encodes her input as $\phi_b := b \cdot \frac{\pi}{2}$, uniformly samples $r_A  \overset{\$}{\leftarrow} \{0, 1\}$, computes the angle $\delta$ (See, Eq.~\ref{eq:delta}) and sends $\delta$ to Bob.
    \item Bob performs the following entangling operations on his private input register $\Ket{\psi_{in}}$, the state $\ket{\psi_A}$, and the ancillary register $\Ket{\psi_B}$: $ (\mathbb{I} \otimes CZ) (CZ \otimes \mathbb{I}) (\ket{\psi_{in}} \otimes \ket{\psi_A} \otimes \Ket{\psi_B})$, where $\ket{\psi_B}$ is in the state $\ket{+}$.
    \item Bob performs the measurement of first register in the X-basis and the second register in the (X,Y)-plane with an angle $\delta$ to obtain the measurement outcomes $m_0\in \{0, 1\}$ and $m_1 \in \{0, 1\}$. 
    \item Bob applies $ X^{m_1} Z^{m_0}$ to the resulting quantum state to obtain:
        \begin{equation}
        \ket{out_b} = X^{\theta_1  \, \oplus \,  r_A  \, \oplus \, m_0\cdot b} [HR_z(-\phi_b)H \ket{\psi_{in}}] 
        \end{equation}
    \item Bob measures $\Ket{out_b}$ in the computational basis and obtains a measurement outcome $\bar{s_b}$ and sends $(m_0, \bar{s_b})$ to Alice. 
    \item  Alice computes (efficiently) $\theta_1$ from $m_{qf}$ using the trapdoor $t_k$ and performs the following (classical) operation to get her desired outcome:  $s_b \coloneqq \bar{s_b} \, \oplus \, \theta_1 \, \oplus \, r_A \, \oplus \, m_0 \cdot b$.
\end{enumerate}
\end{breakablealgorithm}

\begin{theorem}[Correctness] \label{thm:correctness_1_2_oqfe_semihonest}
In an honest run of \onetwooqfe{} Protocol~\ref{protocol:quot_1_2_hbc_alice}, \onetwooqfeprotone{}, when both parties follow  protocol specifications, Alice obtains the outcome $s_b = M_Z [R_x\left(-b \cdot \frac{\pi}{2} \right) \ket{\psi_{in}}]$, where $b$ is Alice's input and $\ket{\psi_{in}}$ is Bob's input.
\end{theorem}

\noindent(\textit{Proof Sketch}).
Firstly, since Alice and Bob run a classical-client remote state preparation ($\sf{RSP}$) (Protocol~\ref{protocol:rsp}), we can equivalently write the Step 1 in Protocol~\ref{protocol:quot_1_2_hbc_alice} as Alice choosing a random bit $\theta_2$, preparing and sending a quantum state $\Ket{\psi_A}$ (Eq.~\ref{eq:rspstate}) to Bob via a quantum channel. This equivalence follows from the correctness of $\sf{RSP}$ protocol (See App.~\ref{app_B_qfactory}, Thm. 3.1.\cite{cojocaru2019qfactory}). Since Controlled-Z ($\sf{CZ}$) operation commutes with the Z-rotations, Step 2 - 5 can be seen as if Bob is effectively applying a $R_z(\theta-\delta)$ on his private input state ($\ket{\psi_{in}}$) to obtain $X^{m_1}HR_z(\theta-\delta)X^{m_0}H\ket{\psi_{in}}$. The interleaved Pauli-X controlled on measurement outputs are by-products of the two measurements performed on Bob's end. The key idea behind this simulation can be seen via two-consecutive runs of a simple circuit identity, also known as one-bit teleportation. Finally, Bob performs computational basis measurement on this final state to obtain $\bar{s}_b$, which Alice updates to obtain $s_b$. The role of $r_A$ (in Eq.~\ref{eq:delta}) is simply to mask the final measurement outcome (from Bob) and is unmasked by Alice in the final step, hence it does not affect the correctness.  A formal proof is presented in Appendix~\ref{app:correctness}.

\begin{theorem}[Simulation-based statistical security against semi-honest Alice]\label{thm:simulation_hbc_alice}
The \onetwooqfe{} Protocol, \onetwooqfeprotone{}, (Protocol~\ref{protocol:quot_1_2_hbc_alice}) securely computes \onetwooqfefunct{} in the presence of semi-honest adversary Alice.
\end{theorem}

\begin{theorem}[Privacy against Malicious Bob]\label{thm:privacy_against_bob}
The \onetwooqfe{} Protocol~\ref{protocol:quot_1_2_hbc_alice} \onetwooqfeprotone{} is private against malicious Bob.
\end{theorem}
 The proofs of Thm~\ref{thm:simulation_hbc_alice} and Thm.~\ref{thm:privacy_against_bob} can be found in Appendix~\ref{App_D_proofs_Sec_5}.

 \subsection{\texorpdfstring{\onetwooqfe}{}: Malicious Alice} \label{sec:maliciousalice}

One of the reasons why the previous Protocol~\ref{protocol:quot_1_2_hbc_alice} is not secure against malicious Alice is that there is no guarantee that the key $k$ sent in the first round represents a valid key. For example, Alice could i) generate $k$ using an algorithm different from $Gen_\mathcal{F}$, ii) run  $Gen_\mathcal{F}$ with some \emph{bad randomness} (a non-uniformly chosen random string).
To circumvent the first problem, one can modify the protocol such that Alice could prove that $k$ is a ``correct" key computed using the algorithm $Gen_\mathcal{F}$ via a zero-knowledge protocol. From the zero-knowledge property, one can ensure that trapdoor of $k$ is not leaked while at the same time Bob can be convinced about the validity of $k$. Unfortunately, this does not prevent Alice from using ``bad'' randomness. To tackle the second problem, we let Alice and Bob engage in a sort of coin-tossing protocol. We denote our protocol enhanced with the coin-tossing and the zero-knowledge proof with \onetwooqfeprottwo{} and refer 
to App.~\ref{app:onesided} for its formal description and proofs.

\section{One-sided simulation \texorpdfstring{\oqfe{}}{} \label{sec:oqfe_protocol}}
\label{sec:oqfe}

In this section we describe a protocol for general two-party quantum computation over classical channels. We call it \oqfesmall{} protocol and we will show that it realizes $\mathcal{F}_{\sf{CQ}}$ in the one-sided simulation security paradigm. 
\sloppy The \oqfe{} protocol, denoted as $\pi_{\sf{Q2PC}}$, represents a generalization of the \onetwooqfe{} construction and similar to the Protocol \ref{protocol:quot_1_2_malicious_alice}, we require the following primitives:
\begin{enumerate}
	\item A commitment scheme ${\sf COM} = (Com, Dec)$ that is hiding against quantum adversaries and computationally binding.
	\item A trapdoor one-way function $\mathcal{F} = (Gen_{\mathcal{F}}, Eval_{\mathcal{F}}, Inv_{\mathcal{F}})$ for the construction of \rsp, which in turn is used as a subroutine in \oqfe{} protocol.
	\item An argument of knowledge post-quantum zero-knowledge protocol $\Pi^\star:=(\provzk^\star,\verzk^\star)$ for the NP-relation: $\Rel=\{com, (dec,m):  \textsf{Dec}(com,dec,m)=1\}$.
	\item An argument of knowledge post-quantum zero-knowledge protocol $\Pi:=(\provzk,\verzk)$ for the NP-relation $\Rel_f$ (defined in Section~\ref{sec:maliciousalice}).
	\item An argument system post-quantum zero-knowledge protocol $\Pi:=(\provzk',\verzk')$ for the NP-relation: 	
	$Rel' = \{(\delta,\pi,s^Z,s^X,com), (r, \theta, \phi',dec): 
	 \delta = \phi' + \theta + r \pi\text{ and }\phi' = (-1)^{s^X} \phi + s^Z \pi\text{ and } \textsf{Dec} (com, dec, \phi) = 1 \}$.
\end{enumerate}

Additionally, we require the universal blind quantum computation (\ubqc{}) protocol with classical output (Protocol 2 of \cite{broadbent2009universal}) as a sub-module. The \ubqc{} protocol is interactive and is based on the measurement-based model of quantum computation. In this model, one can represent an arbitrary quantum function $f$ equivalently as a tuple $(\mathcal{G}_{n \times m}, \Phi, g)$ where $\mathcal{G}_{n \times m}$ is a highly entangled quantum state (often represented as a graph with dimension $(n,m)$ and is also known as graph state), a sequence of classical angles: $\Phi := \{\phi_{i, j}\}$ for $i \in [n]$  and $j \in [m]$, where $\phi_{i, j} \in \{0, \frac{\pi}{4}, \cdots, \frac{7\pi}{4}\}$, and $g$ denotes a set of bits dictating the dependency sets ($s_{i, j}^X$ and $s_{i, j}^Z$), known to both parties, which are required to perform certain Pauli corrections to obtain the desired deterministic computation. We assume that the bits corresponding to $g$ are known to both Alice and Bob and for our purposes we can as well ignore it. Similarly, we can fix the graph state $\mathcal{G}_{n \times m}$, except its dimension $(n,m)$, to say brickwork state~\cite{broadbent2009universal} or cluster state~\cite{raussendorf2003measurement} as both of them are known to be universal for quantum computation with $(X,Y)$-plane measurements~\cite{raussendorf2003measurement,mantri2017universality}. 

At a high level, the protocol can be described as follows. Alice and Bob first run a coin-tossing type of protocol (similar to Sec~\ref{sec:maliciousalice}) to ensure that Alice's randomness was correctly generated and that the \rsp{} primitive (which will be called several times in parallel) is run honestly by Alice. The idea behind using the \rsp{} primitive is to eliminate the need for quantum communication between Alice and Bob, which is crucial in the \ubqc{} protocol. More precisely, using \rsp, Alice prepares the graph state on Bob's side and then Bob entangles his private input as the first layer onto this graph state. Following this, Alice and Bob classically interact (similar to \ubqc{} protocol), wherein each round Alice sends a measurement angle to Bob, where these angles encode Alice's private input. Upon receiving the measurement angles, Bob performs the (projective) measurement in the (X,Y)-plane with that angle and sends the measurement outcome to Alice. This process lasts until all the qubits in the graph are measured. During these rounds, Alice and Bob will also run proof of knowledge systems to ensure the correctness of both parties.

\medskip

\begin{breakablealgorithm}
\caption{\oqfe{} Protocol, $\pi_{\sf{Q2PC}}$, with Classical Alice and Quantum Bob}
\label{protocol:oqfe_protocol}
\vspace{\baselineskip}
\noindent \textbf{Inputs:}
\begin{enumerate}
    \item Sender (Bob): an $n$-qubit state $|\psi_{in}\rangle$. 
    \item Receiver (Alice): $f$ an $n$-qubit unitary represented as the set of angles $\Phi := \{\phi_{i, j}\}_{i, j}$ of a one-way quantum computation over a brickwork state/cluster state~\cite{mantri2017universality}, of the size $n \times m$, along with the dependencies X and Z obtained via flow construction~\cite{danos2006determinism}.
\end{enumerate}
\begin{enumerate}[leftmargin=*]
\item \textbf{Preliminary phase}
\begin{enumerate}
   
    \item[1.1] Alice samples uniformly at random $r_{f, i, j}^A \leftarrow \{0, 1\}^{\lambda}$ and  $r_{i, j} \leftarrow \{0, 1\} $ for $i \in [n]$ and $j \in [m]$.  
    \item[1.2] For each $i \in [n]$ and $j \in [m]$, Alice computes $\textsf{Com}(r^{(i,j)}_{f^A}) \rightarrow (com^{(i,j)}_f, dec^{(i,j)}_f)$ and $\textsf{Com}(\phi_{i,j}) \rightarrow (com^{(i,j)}, dec^{(i,j)})$ and 
  sends $(com^{(i,j)}_f, com^{(i,j)})$ to Bob.
  \item[1.3] For each $i \in [n]$ and $j \in [m]$, Alice runs $\provzk^\star$ on input the statement $x:=com^{(i,j)}$ and the witness $(dec^{(i,j)}, \phi_{i,j})$, and Bob runs $\verzk^\star$ on input the statement $x$. If $\verzk^\star$ outputs $0$ then Bob stops, otherwise he continues with the following steps.
    \item[1.4]  For each $i \in [n]$ and $j \in [m]$, Bob samples $r_{f,i,j}^B$ uniformly at random from $\{0, 1\}^{\lambda}$.  We denote by $r_{f}^B := \{r_{f, i, j}^B\}_{i, j}$. Bob sends $r_{f}^B$ to Alice.
    \item[1.5] Alice computes $r_{f,i,j} = r_{f,i,j}^A \oplus r_{f,i,j}^B$. She then runs $Gen_{\mathcal{F}}$ $n \cdot m$ times using internal random coins $r_{f,i,j}$ and obtains $(k^{(i,j)}, t_k^{(i,j)}, hp^{(i,j)})$  for $i \in [n]$ and $j \in [m]$. Denote by $k := \{k^{(i,j)}\}_{i, j}$ is the concatenation of the $n \cdot m$ public keys. 
  \item[1.6] For each $i \in [n]$ and $j \in [m]$:
   \begin{itemize}
  	  \item[1.6.1] Alice runs $\provzk$ on input $x^{(i,j)} = (k^{(i,j)}, r^{(i,j)}_{f^B}, com^{(i,j)}_f), w = (r^{(i,j)}_{f^A}, dec^{(i,j)}_f)$ and Bob runs $\verzk$ on input $x^{(i,j)}$.
 	
  	  \item[1.6.2] If $\verzk$ outputs $0$ then Bob aborts. Otherwise, he continues to the next step.
     \end{itemize}
\end{enumerate}
\end{enumerate}

\begin{enumerate}[resume]
\item \textbf{QFactory and UBQC}
    \begin{enumerate}

    \item[2.1] For each $i \in [n]$ and $j \in [m]$, Alice on input  $t_k^{(i,j)}$, and Bob on input $k^{(i,j)}$ run an instances of 8-states QFactory protocol\footnote{An 8-states QFactory protocol is combination of two runs of 4-states QFactory Protocol~\ref{protocol:rsp} given in~\cite{cojocaru2019qfactory}. The only difference between the QFactory Protocol~\ref{protocol:rsp}, and the one used in our construction is that Alice does not execute the first step (a), since the key for the trapdoor OWFs has been already generated as described in the previous steps.} (in sequence) to obtain $\theta_{i, j}$ on client's side and $\ket{+_{{\theta}_{i, j}}}$ on server's side, where $\theta_{i, j} \gets \Z \frac{\pi}{4}$, $i \in [n]$, $j \in [m]$. For each qubit $\ket{+_{\theta_{i, j}}}$. 

    \item[2.2] Bob entangles all these qubits by applying controlled-Z gates between them in order to create a graph state $\mathcal{G}_{n \times m}$, where the first layer of the graph is Bob's input $|\psi_{in}\rangle$.
    \item[2.3] For $j \in [m]$ and $i \in [n]$:
\begin{enumerate}
\item[2.3.1] Alice computes $\delta_{i, j} = \phi_{i, j}' + \theta_{i, j} + r_{i, j}\pi$, where $\phi_{i, j}' = (-1)^{s_{i, j}^X} \phi_{i, j} + s_{i, j}^Z \pi$ and $s_{i, j}^X$ and $s_{i, j}^Z$ are computed using the previous measurement outcomes and the X and Z dependency sets. Alice then sends the measurement angle $\delta_{i, j}$ to Bob.
\item [--] Alice runs $\provzk'$ on input the statement to be proven 
$x:=(\delta_{i, j},\pi_{i, j},s_{i, j}^Z,s_{i, j}^X,com_{i, j})$ and the witness $w:=(r_{i, j}, \theta_{i, j}, \phi'_{i, j},dec_{i, j})$,  and Bob runs 
the interactive algorithm $\verzk,$ on input the statement $x$. 
Let $b$ be the output of $\verzk'$. If $b=0$ then Bob aborts, otherwise he continues as follows.
\item[2.3.2] Bob measures the qubit $\ket{+_{{\theta}_{i, j}}}$ in the basis $\{\ket{+_{{\delta}_{i, j}}}, \ket{-_{{\delta}_{i, j}}}\}$ and obtains a measurement outcome $s_{i, j}' \in \{0, 1\}$. Bob sends the updated measurement result $s_{i, j}'$ to Alice.
\item[2.3.3] Alice computes $\bar{s}_{i, j} = s_{i, j}' \oplus r_{i, j}$.
\end{enumerate}
    \end{enumerate}
\end{enumerate}
 
\noindent \textbf{Output:} Alice obtains the output $f(\Ket{\psi_{in}})$ as the concatenation of $\{\bar{s}_{i, m}\}_i$. 
\vspace{0.5\baselineskip}
\end{breakablealgorithm}
\vspace{\baselineskip}

\begin{theorem}[Correctness]
In an honest run of the \oqfe{} Protocol \ref{protocol:oqfe_protocol}, when both parties follow the protocol specifications, Alice obtains the outcome $f(\ket{\psi})$, where $f$ is Alice's input and $\ket{\psi}$ is Bob's input.
\label{thm:correctness_oqfe}
\end{theorem}
\noindent (\textit{Proof Sketch}).
One of the main differences between the protocols for \onetwooqfe{} (Protocol~\ref{protocol:quot_1_2_hbc_alice} and Protocol~\ref{protocol:quot_1_2_malicious_alice}) and \oqfe{} (Protocol~\ref{protocol:oqfe_protocol}) is that, in the former, we consider a simple linear graph state (of three qubits) instead of cluster state (as in the \oqfe{}) along with the fact that Alice's input is encoded using a larger set of angles and Bob's input is multi-qubit state. 
The rest of the steps in the \oqfe{} protocol are straightforward extensions of \onetwooqfe, where we also use standard cryptographic techniques such as zero-knowledge proofs to ensure that the classical messages of both the parties were computed honestly. The \oqfe{} protocol is based on measurement-based model of quantum computing and its correctness follows from the correctness proof of \onetwooqfe{} (\autoref{thm:correctness_1_2_oqfe_semihonest}) as well as the correctness of the UBQC protocol \cite{broadbent2009universal}. 

\begin{theorem}\label{th:oqfeproof}
Protocol $\pi_{\sf{Q2PC}}$ securely computes $\mathcal{F}_{\sf{CQ}}$ with one-sided simulation.
\label{lemma:one_sided_simulation_full_oqfe}
\end{theorem}
To complete the proof we need to prove the following two lemmata: Lemma~\ref{thm:privacy_against_bob_oqfe} and Lemma~\ref{thm:simulation_malicious_alice_oqfe}.

\begin{lemma}[Privacy against Malicious Bob]
The  \oqfe{} Protocol~\ref{protocol:oqfe_protocol} is private against malicious Bob.
\label{thm:privacy_against_bob_oqfe}
\end{lemma}
\begin{proof}

This results from the privacy against Bob of the protocol resulting from
combining the quantum-client UBQC protocol with QFactory proven in Theorem 5.3 of \cite{badertscher2020security} together with the zero-knowledge property of $\Pi$ and $\Pi'$, and the hiding of the commitment scheme.
\end{proof}

\begin{lemma}[Simulation-based Security Malicious Alice] The \oqfe{} Protocol~\ref{protocol:oqfe_protocol} is simulation-based secure against malicious Alice.
\label{thm:simulation_malicious_alice_oqfe}
\end{lemma}
The proof can be found in Appendix~\ref{app:simulation_malicious_alice_oqfe}.

\begin{remark} We emphasize that by instantiating our \oqfe{} protocol with the zero-knowledge proof of knowledge system of \cite{ananth2020concurrent} together with the statistical binding post-quantum hiding commitment scheme of \cite{BDLOP18}, our Protocol~\ref{protocol:oqfe_protocol} is simulation-based secure against \textit{unbounded} Alice.
\end{remark}

\section{Fully-simulatable (Black-Box) \texorpdfstring{\qtwopc{}}{} over classical channel} \label{sec:full_simulation_two_pc}

In this section we will describe a general compiler, $C$, that allows transforming a \qtwopc{} protocol achieving one-sided simulation security into a \qtwopc{} protocol achieving full-simulation security, under the assumption that Bob has a classical description of his input. To construct $C$ we will make use of the following primitives:
\begin{enumerate}
    \item The \oqfe{} Protocol~\ref{protocol:oqfe_protocol}, $\pi_{\sf{Q2PC}}$, with classical inputs, which we call $\sf{OS-Q2PC}$. Using Theorem~\ref{th:oqfeproof}  we know that the protocol $\sf{OS-Q2PC}$ securely computes $\mathcal{F}_{\sf{CC}}$ with one-sided simulation.
    \item Post-quantum Zero-Knowledge post-quantum proof of knowledge $(P_C^{(1)}, V_C^{(1)})$ for $NP$ relation $Rel_C$, which we call \zkpok. 
    \item Classical Zero-Knowledge $(P_Q^{(2)}, V_C^{(2)})$ for \QMA{} relation $Rel_Q$, denoted $\sf{ZK}$.
     \item Classical Commitment scheme, post-quantum hiding and statistical binding, which we call $\sf{PQCOM} = (Com, Dec)$;
\end{enumerate}

The starting point is our Protocol~\ref{protocol:oqfe_protocol}, achieving simulation security against Alice and privacy against Bob in the black-box setting. We will make use of Protocol~\ref{protocol:oqfe_protocol} in a black-box manner. \\

\begin{breakablealgorithm}
\caption{Compiler $C$ for achieving full-simulation \qtwopc{}}
\label{protocol:full_sim_qtwopc}
\vspace{\baselineskip}
\noindent \textbf{Private Inputs:}
\begin{enumerate}
    \item Bob:  $\ket{\psi_{in}}$ with classical description $y_{\ket{\psi}}$; 
    \item Alice: $x \in \{0, 1\}^n$;
\end{enumerate}
\noindent \textbf{Protocol Steps:}
\begin{enumerate}
    \item Bob computes $Com(y_{\ket{\psi}}) \rightarrow (com_{y}, dec_{y})$ and sends $com_y$ to Alice;
    \item Bob and Alice will run $(P_C^{(1)}, V_C^{(1)})$ for the \NP{} relation:
        $$ Rel_C = \{ x = com_{y} , w = (dec_{y}, y_{\ket{\psi}}) \text{ such that $w$ is the decommitment of $x$} \}  $$
    \item Alice and Bob run the one-sided simulation protocol $\sf{OS-Q2PC}$ on Alice's input $x$ and Bob's input  $\ket{\psi_{in}}$. We denote the messages sent by Bob to Alice by $\{m_i\}_{i = 1}^N$;
    \item Bob and Alice will run the post-quantum zero-knowledge $(P_Q^{(2)}, V_C^{(2)})$, for Bob to prove to Alice that for every $i \in [N]$, the message $m_i$ was computed correctly according to the $\sf{OS-Q2PC}$ on Alice's input $x$ and Bob's quantum input corresponding to the value committed in $com_y$;
\end{enumerate}

\end{breakablealgorithm}

\begin{theorem} Protocol~\ref{protocol:full_sim_qtwopc} is a secure black-box \qtwopc{} (as defined in Def.~\ref{def:Qtwopc}), assuming that Bob has a classical description of his input.
\label{thm:full_sim_qtwopc}
\end{theorem}
The proof can be found in Appendix~\ref{app:theorem_full_sim}.

\paragraph{Instantiation}. As before we will instantiate $\sf{PQCOM}$ with the scheme of \cite{BDLOP18} and $(P_C^{(1)}, V_C^{(1)})$ with the construction of \cite{ananth2020concurrent}. For $(P_Q^{(2)}, V_C^{(2)})$ we can use the construction of \cite{vidick_zhang_ZK_20}. \\

\section{\texorpdfstring{\zkpoqk{}}{} Compiler} \label{sec:compiler_zkpoqk}

In this section, we will present a general compiler that constructs post-quantum zero-knowledge classical proof of quantum knowledge for \QMA{} (\zkpoqk) from simpler primitives: classical proof of quantum knowledge, post-quantum commitment scheme, post-quantum zero-knowledge proof of classical knowledge, and post-quantum private-key encryption scheme.  Then, we will show how to instantiate the required primitives for this compiler and in Section~\ref{sec:full_simulation_two_pc} we will show how to use the resulting \zkpoqk{} to obtain full-simulation secure \qtwopc{} over a classical channel.
But first, we will proceed with some definitions required for the construction of our compiler.

\begin{definition}[\mesind] We say a Proof of (Quantum) Knowledge $(P, V)$ satisfies the \mesind{} property if the messages computed by $V$ are independent of the messages received from $P$. 
\end{definition}
Next we propose the compiler for constructing \zkpoqk. This relies on the following primitives:

\begin{itemize}
    \item Classical Proof of Quantum Knowledge $(P_Q^{(1)}, V_C^{(1)})$ for \QMA{} relation $Rel_Q$, denoted \poqk{}, which also must satisfy the \mesind{} property; 
    \item Classical Commitment scheme, post-quantum hiding and post-quantum binding, which we call $\sf{PQCOM} = (Com, Dec)$;
    \item Classical post-quantum Zero-Knowledge post-quantum proof of (classical) knowledge $(P_C^{(2)}, V_C^{(2)})$ for \NP{} relation $Rel_C$, which we call \zkpok{};
    \item Post-quantum private-key encryption scheme $PQE = (Gen, Enc, Dec)$
\end{itemize}

Then, we can achieve the \zkpoqk{} $(P_Q, V_C)$ for the \QMA{} relation $Rel_Q$ using the following compiler: \\

\begin{breakablealgorithm}
\caption{ \zkpoqk{} Protocol, $(P_Q, V_C)$ for the \QMA{} relation $Rel_Q$}
\label{protocol:zkcpoqk}
\begin{enumerate}
    \item $P_Q$ runs $Gen(1^\lambda)$ and obtains the secret key $sk$ for the $PQE$;
    \item $P_Q$ computes $Com(sk) \rightarrow (com_{sk}, dec_{sk})$ and sends $com_{sk}$ to $V_C$; 
    \item $P_Q$ and $V_C$ will run the $(P_C^{(2)}, V_C^{(2)})$ for the \NP{} relation:
        $$ Rel_C = \{ x = com_{sk} , w = (dec_{sk}, sk) \text{ such that $w$ is the decommitment of $x$} \}  $$
    \item Consider the classical proof of quantum knowledge $(P_Q^{(1)}, V_C^{(1)})$ for $Rel_Q$, and the number of messages they exchange is $N$. Then:
    \item For $i \in [N]$:
    \begin{enumerate}
    \item Let $m_i$ be the message $P_Q^{(1)}$ would have sent to $V_C^{(1)}$. Then, $P_Q$ will compute $Enc(sk, m_i) \rightarrow enc_i$ and will send $enc_i$ to $V_C$;
    \item Let  $m_i'$ be the message that $V_C^{(1)}$ would have sent to $P_Q^{(1)}$, $V_C$ sends exactly the same message $m_i'$ to $P_Q$;
    \end{enumerate}
    \item $P_Q$ and $V_C$ will run a post-quantum zero knowledge $(P_C^{(3)}, V_C^{(3)})$, for $P_Q$ to prove to $V_C$ that for all $i \in [N]$, $enc_i$ were computed using the secret key $sk$, which is committed in $com_{sk}$, and that their decryptions were computed according to the $(P_Q^{(1)}, V_C^{(1)})$ protocol and $V_C^{(1)}$ would have accepted the resulting transcript;
    \item If $V_C^{(1)}$ and $V_C^{(2)}$ accept, then $V_C$ outputs $1$, otherwise outputs $0$.
\end{enumerate}
\end{breakablealgorithm}
\vspace{\baselineskip}

We want to show that the proposed construction Protocol~\ref{protocol:zkcpoqk} has the post-quantum zero-knowledge property and that it will inherit the classical proof of quantum knowledge property of ($P_Q^{(1)}, V_C^{(1)}$).


\begin{theorem} Assume there exists a \mesind{} \poqk{} $(P_Q^{(1)}, V_C^{(1)})$ with knowledge error $\kappa_1$ and a \zkpok{} $(P_C^{(2)}, V_C^{(2)})$ with knowledge error $\kappa_2$. Then there exists a \zkpoqk{} $(P_Q, V_C)$ with knowledge error $\kappa_1 \kappa_2$. Moreover, if $(P_C^{(2)}, V_C^{(2)})$ and $(P_Q^{(1)}, V_C^{(1)})$ are simulatable then $(P_Q, V_C)$ is simulatable.
\label{thm:zkpoqk}
\end{theorem}

The proof of Theorem~\ref{thm:zkpoqk} can be found in Appendix~\ref{App_G_proofs_Sec_8}. 









\paragraph{Instantiating \zkpoqk.} In this part we will show how can we instantiate our general compiler and determine the properties/parameters of the resulting \zkpoqk{} construction.
For the \poqk{} we can instantiate the proof of quantum knowledge $(P_Q^{(1)}, V_C^{(1)})$ with the construction of \cite{vidick2020classical} which works for relations from a subset of \QMA, called $\QMA^*$ (see Def.~\ref{def:qma_assump}).  For the post-quantum commitment scheme we can employ the scheme of \cite{BDLOP18}. For the \zkpok{} $(P_C^{(2)}, V_C^{(2)})$ we will employ again the construction of \cite{ananth2020concurrent}. Finally, for the Post-quantum private-key encryption scheme $PQE$ we can use the construction given in \cite{regev2009lattices}.

\paragraph{Acknowledgments}
The authors thank Thomas Vidick and anonymous reviewers for their constructive comments, which helped us to improve the manuscript. MC acknowledges support from H2020 project PRIVILEDGE \#780477. AC acknowledges support from the French National Research Agency through the project ANR-17- CE39-0005 quBIC. EK acknowledges support from the EPSRC Verification of Quantum Technology grant (EP/N003829/1), the EPSRC Hub in Quantum Computing and Simulation (EP/T001062/1), grant FA9550-17-1-0055, and the UK Quantum Technology Hub: NQIT grant (EP/M013243/1) and funding from the EU Flagship Quantum Internet Alliance (QIA) project. AM gratefully acknowledges funding from the AFOSR MURI project ``Scalable Certification of Quantum Computing Devices and Networks’’. This work was partly done while AM was at the University of Edinburgh, UK where it was supported by EPSRC Verification of Quantum Technology grant (EP/N003829/1).

\bibliographystyle{abbrv}
\bibliography{ref}




\newpage
\begin{center}
  \LARGE Supplementary Material  
\end{center}
\appendix


\section{Computational Complexity Classes and Cryptographic Primitives} \label{app:cp}
 \label{app:ct}

We provide some well-known relations and the class of languages associated with it.  Classically, a relation over finite sets $\{0,1\}^{*} \times \{0,1\}^{*}$ is a subset $R\subseteq \{0,1\}^{*} \times \{0,1\}^{*}$, and the language associated with R is $L_R=\{x: \exists y : (x,y) \in R\}$. 

 \begin{definition}[$\NP$] 
 The class $\NP$ consists of all languages $L \subseteq \{0,1\}^{*}$ for which there  exists a uniformly generated family of classical, deterministic, poly-size circuits $\{V_{x}: x\in\{0,1\}^{*}\}$ and a polynomial $m$, such that the following holds:
 \begin{enumerate}
     \item (Completeness) For all $x \in L$ there exists an $m(|x|)$-bit witness $w$ such that $V_x(w) = 1$ 
     \item (Soundness) For all $x \notin L$ and for all $m(|x|)$-bit witness $w$, $V_x(w) = 0$. 
 \end{enumerate}
 \end{definition}
 
 \begin{definition}[$\MA$]
 The class $\MA$ consists of all languages $L \subseteq \{0,1\}^{*}$ for which there  exists a uniformly generated family of classical, randomized, poly-size circuits $\{V_{x}: x\in\{0,1\}^{*}\}$ and a polynomial $m$, such that the following holds:
 \begin{enumerate}
     \item (Completeness) For all $x \in L$ there exists an $m(|x|)$-bit witness $w$ such that $\Pr(V_x(w) = 1) \geq 2/3$ 
     \item (Soundness) For all $x \notin L$ and for all $m(|x|)$-bit witness $w$, $\Pr(V_x(w) = 0) \geq 1/3$. 
 \end{enumerate}
 \end{definition}

\begin{definition}[$\MA$-relation]
A relation R is an MA-relation if there is a $\ppt$ Verifier V such that:
\begin{enumerate}
    \item (Completeness) $(x,w)\in L_R \implies \Pr [V_{|x|}(x,w) = 1] \geq 2/3 $
    \item (Soundness) $x \notin L_R \implies \Pr [V_{|x|}(x,w) = 1] \leq 1/3.$
\end{enumerate}
where $ V = \{V_n\} $ are the uniformly generated family of circuits. 
\end{definition}

In the quantum case we replace the ``witness'' $w$ (the second argument) with a quantum state $\ket{\psi}$ and define the class \QMA{} with  polynomial-size quantum circuits $Q=\{Q_n\}_{n\in \N}$ such that for every $n$, $Q_n$ takes as input a string $x\in\{0,1\}^n$ and a quantum state $\sigma$ on $p(n)$ qubits (for some polynomial $p(n)$) and returns a single bit as output. 

 \begin{definition}[\QMA]
 \label{def:qma}
 The class \QMA{} consists of all languages $L \subseteq \{0,1\}^{*}$ for which there  exists a uniformly generated family of quantum poly-size circuits $\{Q_{x}: x\in\{0,1\}^{*}\}$ and a polynomial $m$, $p$ where each $V_x$ has $m(|x|)$ input qubits, $k(|x|)$ auxiliary qubits and its output is given by the first output qubit such that the following holds:
 \begin{enumerate}
     \item (Completeness) For all $x \in L$ there exists an $m(|x|)$-qubit witness $\ket{\psi}$ such that $\Pr(V_x \text{ accepts } \ket{\psi}) \geq 2/3$ 
     \item (Soundness) For all $x \notin L$ and for all $m(|x|)$-qubit witness $\ket{\psi}$, $\Pr(V_x \text{accepts} \ket{\psi}) \leq 2/3$.
 \end{enumerate}
 \end{definition}
Note that the completeness and soundness can be amplified to $1 - 2^{-poly(|x|)}$ and $2^{-poly(|x|)}$, respectively.

\begin{definition}[\QMA-relation]
  \label{def:qma-relation}
  Let $A$ be a problem in \QMA{} (See Definition~\ref{def:qma}), and let $Q$ be a $\qpt$ verifier, with completeness $\alpha$ and soundness $\beta$. Then, we say that $R_{Q, \gamma}$ is a \QMA-relation such that the following holds
    \begin{enumerate}
    \item (Completeness) $(x,\ket{\psi}) \in R_{Q,\alpha} \implies \Pr [Q_{|x|}(x,\ket{\psi}) = 1] \geq \alpha $
    \item (Soundness)$(x,\rho) \not\in R_{Q,\beta} \implies \Pr [Q_{|x|}(x,\rho) = 1] < \beta.$
\end{enumerate}
  \end{definition}
  
\begin{definition}[$\QMA^*$, \cite{vidick2020classical}]
\label{def:qma_assump}
Let $(Q,\alpha,\beta)$ be a \QMA{} relation. We define $\QMA^*$ as a relation that satisfies the following properties: 
\begin{enumerate}
\item The completeness parameter $\alpha$ is negligibly close to $1$, and the soundness parameter $\beta$ is bounded away from $1$ by an inverse polynomial. 
\item For any $x\in \{0,1\}^n$ there is a local Hamiltonian $H=H_x$ that is efficiently constructible from $x$ and satisfies the following. First, we assume that $H$ is expressed as a linear combination of tensor products of Pauli operators with real coefficients chosen such that $-\mathrm{Id} \leq H\leq \mathrm{Id}$. Second, whenever there is $\sigma$ such that $(x,\sigma)\in R_{Q,\alpha}$, then  $\Tr(H\sigma)$ is negligibly close to $-1$ and moreover any $\sigma$ such that  $\Tr(H\sigma)\leq -1 + \delta$ satisfies $\Pr(Q_{|x|}(x,\sigma)=1)\geq 1-r(|x|)q(\delta)$ for some polynomials $q,r$ depending on the relation only. Third, whenever $x\in N_{Q,\beta}$ then the smallest eigenvalue of $H$ is larger than $-1+1/s(|x|)$, where $s$ is another polynomial depending on the relation only. 
	\end{enumerate}
\end{definition}

\paragraph{Interactive quantum machines~\cite{unruh2012quantum,coladangelo2020non}.} An \textit{interactive quantum machine} is a machine $M$ with two registers: a register $\textsf{S}$ for its internal state, and a register $\textsf{N}$ for sending and receiving messages (the network register). Upon activation, $M$ expects in $\textsf{N}$ a message, and in $\textsf{S}$ the state at the end of the previous activation. At the end of the current activation, $\textsf{N}$ contains the outgoing message of $M$, and $\textsf{S}$ contains the new internal state of $M$. A machine $M$ gets as input: a security parameter $\mu \in \mathbb{N}$, a classical input $x \in \{0,1\}^*$, and quantum input $\ket{\Phi}$, which is stored in $\textsf{S}$. Formally, machine $M$ is specified by a family of circuits $\{M_{\mu x}\}_{\mu \in \mathbb{N}, x \in \{0,1\}^*}$, and a family of integers $\{r_{\mu x}\}_{\mu \in \mathbb{N}, x \in \{0,1\}^*}$. $M_{\mu x}$ is the quantum circuit that $M$ performs on the registers $\textsf{S}$ and $\textsf{N}$ upon invocation. $r_{\mu x}$ determines the total number of messages/invocations. We might omit writing the security parameter when it is clear from the context. We say that $M$ is \textit{quantum-polynomial-time} (QPT) if the circuit $M_{\mu x}$ has polynomial size in $\mu + |x|$, the description of the circuit is computable in deterministic polynomial time in $\mu + |x|$ given $\mu$ and $x$, and $r_{\mu, x}$ is polynomially bounded in $\mu$ and $x$.

Usually, both these registers are assumed to be quantum but for this work, we require the register $\textsf{N}$ to be strictly classical. In this work, we model one of the parties (Bob) as a quantum interactive machine while the other party (Alice) is only required to perform classical (stochastic) operations. We will denote the interactive machines for Alice and Bob as $\alice$ and $\bob$, respectively, with internal registers $\textsf{S}$ and $\textsf{S}'$  and the classical network register as $\textsf{N} $ and $\textsf{N}'$, respectively.  Finally, we assume that Alice (completely classical party) sends the first message as well as receives the last message.

\begin{definition}[k-regular]\label{def:k_regular} A deterministic function $f \colon \mathcal{D} \rightarrow \mathcal{R}$ is \textbf{k-regular} if $ \, \, \forall y \in \Ima f$, we have ${|f^{-1}(y)| = k}$.
\end{definition}

\begin{definition}[Trapdoor One-Way Function]\label{def:trapdoor_function} 
  A family of functions $\{f_k : \mathcal{D} \rightarrow \mathcal{R} \}$
   is a \textbf{trapdoor function} if:
\begin{itemize}
\item There exists a PPT algorithm {\tt Gen} which on input $1^n$ outputs $(k, t_k)$, where $k$ represents the index of the function.
\item $\{f_k : \mathcal{D} \rightarrow \mathcal{R}\}_{k \in \mathcal{K}}$ is a family of one-way functions, namely:
\begin{itemize}
\item There exists a PPT algorithm that can compute $f_k(x)$ for any index $k$, outcome of the PPT parameter-generation algorithm \text{Gen} and any input $x \in \mathcal{D}$;
\item Any QPT algorithm $\cA$ can invert $f_k$ with at most negligible probability over the choice of $k$: \\
  $ \underset{\substack{k \leftarrow Gen(1^n) \\  x \leftarrow \mathcal{D} \\ rc \leftarrow \{0, 1\}^{*}}} \Pr [f(\mathcal{A}(k, f_k(x)) = f_k^{-1}(x)] \leq \negl$\\
  where $rc$ represents the randomness used by $\mathcal{A}$
\end{itemize}
\item There exists a PPT algorithm {\tt Inv}, which on input $t_k$ (which is called the trapdoor information) output by {\tt Gen}($1^n$) and $y = f_k(x)$ can invert $y$ (by returning all preimages of $y$\footnote{While in the standard definition of trapdoor functions it suffices for the inversion algorithm {\tt Inv} to return one of the preimages of any output of the function, in our case we require a two-regular trapdoor function where the inversion procedure returns both preimages for any function output.})
with overwhelming probability over the choice of $(k, t_k)$ and uniform choice of $x$.
\end{itemize}
\end{definition}

\textbf{Instantiation.} A trapdoor one-way function can be instantiated from the construction of \cite{MP12} and a 2-regular variant can be found in \cite{cojocaru2019qfactory}.

\begin{definition}[Hardcore Predicate]\label{def:hardcore_predicate}
  A function $hc \colon \mathcal{D} \rightarrow \{0, 1\}$ is a \textbf{hardcore predicate} for a function $f$ if:
  \begin{itemize}
\item There exists a PPT algorithm that, for any input $x$, can compute $hc(x)$;
\item Any QPT algorithm $\mathcal{A}$ when given $f(x)$, can compute $hc(x)$ with negligible better than $1/2$ probability: \\
$ \underset{\substack{x \leftarrow \mathcal{D}(n) \\ rc \leftarrow \{0, 1\}^{*}}} \Pr [\mathcal{A}(f(x), 1^n) = hc(x)] \leq \frac{1}{2} + \negl$, where $rc$ is the randomness used by $\mathcal{A}$;
\end{itemize}
\end{definition}

\begin{definition}[Commitment Scheme] $CS = (Sen, Rec)$ is a 2-phase protocol between 2 polynomial-time interactive algorithms: sender $Sen$ and receiver $Rec$. In the \textit{commitment phase} $Sen$ with input $m$ interacts with $Rec$ to produce a commitment $com$ and the private output $d$ of $Sen$.
\begin{enumerate}
    \item \textbf{Correctness}: On the \textit{decommitment phase}, $Rec$ on input $m$ and $d$ accepts $m$ as decommitment of $com$.
    \item \textbf{Computational (post-quantum) Hiding}: For any QPT adversary $Rec^*$ interacting with $Sen$, the probability distributions describing the output of $Rec^*$: $\{\langle Sen(0), Rec^* \rangle \}$ and $\{\langle Sen(1), Rec^* \rangle \}$ are computationally indistinguishable.
    \item \textbf{Statistical Binding}: For any commitment $com$ generated during the commitment phase by a malicious unbounded sender $Sen^*$, there exists negligible $negl$ such that $Sen^*$ with probability at most $negl$ outputs 2 decommitments $(m_0, d_0)$ and $(m_1, d_1)$ with $m_0 \neq m_1$ such that $Rec$ accepts both decommitments.
\end{enumerate}
For \textbf{non-interactive commitment scheme} $(Com, Dec)$ we use the notation:
\begin{enumerate}
    \item Commitment phase: $Com(m) \rightarrow (com, dec)$, where $com$ is the commitment of the message $m$ and $dec$ is the corresponding decommitment information.
    \item Decommitment phase: $Dec(com, dec, m) = 1$
\end{enumerate}
\end{definition}

\textbf{Instantiation.} A post-quantum hiding, statistical binding commitment scheme can be instantiated from the scheme proposed in \cite{BDLOP18}.

\begin{definition}[Proof/Argument system] A pair of \ppt\ interactive algorithms $\Pi=(\prover,\verifier)$ constitutes a {\em proof system} (resp., an {\em argument system}) for an \NP-language $L$, if the following condition holds:
\begin{description}
	\item[Completeness:] For every $x\in L$ and $w$ such that $(x,w)\in\Rel_\mathsf{L}$, it holds that:
	\[
		\Pr[\langle \prover(w), \verifier \rangle (x) = 1]=1.
	\]
	\item  [Soundness:] For every interactive (resp., PPT interactive) algorithm $\prover^{\star}$, 
there exists a negligible function $\nu$ such that for every  $x \notin L$ and every $z$:
$$\Pr[\langle \prover^{\star}(z), \verifier \rangle (x) =1]< \negl.
$$
	\end{description}
\end{definition}

\begin{definition}[Proof of Knowledge\footnote{This definition can be easily extended to the CRS model. In that case the extractor would have the additional power of programming the CRS. We make this algorithm explicit when it is required for our constructions.}]\label{def:pok}
A pair $(\Prov,\Ver)$ of $\ppt$ interactive machines is a {\em proof of knowledge with  knowledge error} $k(\cdot)$ for polynomial-time relation $\Rel$ if the following properties hold:
\begin{itemize}[topsep=0pt]
\item {\em Completeness:}
for every $(x,w)\in\Rel$, it holds that $$\prob{\langle\Prov(w),\Ver\rangle(x)=1}=1-\mathsf{negl}(|x|).$$
\item \emph{Knowledge Soundness:} there exists a probabilistic oracle machine $\Extract$, called the {\em extractor}, running in expected probabilistic polynomial time, such that for every interactive machine $\proverstar$ and for every input $x$ accepted by $\Ver$ when interacting with $\proverstar$ with probability  $\epsilon(x) > k(x)$, we have
\[\Pr\left(  \left((x, w ) \in R \right): w \leftarrow \Extract^\proverstar(x) \right) \geq p\left(\epsilon(x)-k(x), \frac{1}{\poly(|x|)}\right). \]
\end{itemize}
\end{definition}

The notion of an argument of knowledge is essentially the same but it requires the knowledge soundness property to hold against PPT adversaries and for a sufficiently long input~\cite{DBLP:conf/crypto/BellareG92}.

\noindent Recently, proof systems have also been extended to \QMA-relations in~\cite{broadbent2019zero,coladangelo2020non}. The main difference from Quantum Proofs of (classical) Knowledge is that in the case of \QMA{} relations, the notion of a witness is in a different manner than $\NP$ relations. For any $0\leq \gamma \leq 1$, a quantum relation is defined as follows:
\[R_{Q,\gamma} = \{(x,\sigma) :  Q\text{ accepts } (x, \sigma) \text{ with probability at least } \gamma\}.
\]
The parameter $\gamma$ quantifies the expected success probability for the verifier and roughly speaking, $\gamma$ is a measure of the ``quality'' of a witness $\ket{\psi}$ (or a mixture thereof, as represented by the density matrix $\sigma$) that is sufficient for the witness to be acceptable for the relation $R$. 

\begin{definition}[Proof of Quantum Knowledge~\cite{broadbent2019zero,coladangelo2020non}]
  \label{def:poqk}
  Let $R_{Q,\gamma}$ be a \QMA{} relation.
  A proof system $(P,V)$ is a Proof of Quantum Knowledge for $R_{Q,\gamma}$ with knowledge error $\kappa(n) > 0$ and quality $q$, if there exists a polynomial $p > 0$ and a  polynomial-time machine $\Extract$ such that for any quantum interactive machine $P^*$ that makes $V$ accept some instance $x$ of size $n$ with probability at least $\eps > \kappa(n)$, we have:
  \[Pr\left[ \left( (x, \sigma) \in R_{Q,q(\eps,\frac{1}{n})} \right): \sigma \leftarrow \Extract^{\ket{P^*(x,\rho)}}(x) \right]  \geq
  p\left(\eps-\kappa(n),\frac{1}{n}\right).
  \]
\end{definition}

Next, we describe generalised proofs of quantum knowledge (or quantum agree-and-prove (AaP) schemes). In an AaP protocol, the auxiliary inputs that the prover and the verifier receive before they begin interacting are captured using the  \emph{input generation algorithm}. Informally, the input generation algorithm models `prior knowledge' which the prover and the verifier may possess.
 \begin{definition}[Input Generation Algorithm,~\cite{vidick2020classical}]
\label{def:input-gen}
An input generation algorithm $I$ for an agree-and-prove scenario $\mathcal{S}$ is a machine $I$ taking a unary encoding of the security parameter $\lambda$ as input and producing a CQ state $\rho_{\reg{AUX}_V \reg{AUX}_P}$ specifying the auxiliary inputs for the verifier (in the classical register $\reg{AUX}_V$) and prover (in the quantum register $\reg{AUX}_P$) respectively as output. We may use the shorthand $\rho_{\reg{AUX}_P} \equiv \Tr_{\reg{AUX}_V}\big( \rho_{\reg{AUX}_V \reg{AUX}_P} \big)$ and $\rho_{\reg{AUX}_V} \equiv \Tr_{\reg{AUX}_P}\big( \rho_{\reg{AUX}_V \reg{AUX}_P} \big)$.
\end{definition}

 \begin{definition}[Agree-and-Prove Scenario for quantum relations,~\cite{badertscher2019agree,vidick2020classical}]
An agree-and-prove (AaP) scenario for quantum relations is a triple $(\cF,\cR,\cC)$ of interactive oracle machines satisfying the following conditions: 
\begin{enumerate}
\item The \textbf{setup functionality} $\cF$ is a QPT ITM taking a unary encoding of a security parameter $\lambda$ as input. The ITM $\cF$ runs an initialization procedure \texttt{init}, and in addition returns the specification of an oracle (which we also model as an ITM) $\cO_\cF(i,\texttt{q},arg)$. The oracle function takes three arguments: $i\in \{I,P,V\}$ denotes a `role', $\texttt{q}$ denotes a keyword specifying a query type, and $arg$ denotes the argument for the query. There are three different options for the `role' parameter, which exists to allow $\cF$ to release information selectively depending on the party asking for it. The roles $I$, $P$ and $V$ correspond respectively to the \emph{input generator} (Definition \ref{def:input-gen}), the prover, and the verifier.

\item The \textbf{agreement relation} $\cC$ is a QPT oracle machine taking a unary encoding of the security parameter $\lambda$ and a statement as inputs, and producing a decision bit as output.

\item The \textbf{proof relation} $\cR$ is a QPT oracle machine taking a unary encoding of the security parameter $\lambda$, a (classical) statement $x$ and a (quantum) witness $\rho_{W}$ as inputs, and outputting a decision bit.
\end{enumerate}
\label{def:AaP}
\end{definition}

One can model proofs and arguments of knowledge using the AaP scenarios, where the agreement phases are trivial and \NP{} or a \QMA{} relation are the proof relation. 

 \begin{definition}[Nondestructive interaction,~\cite{vidick2020classical}]
\label{def:nondestructive-interaction}
Let $P = (\{P_{\lambda x}\}, \{r^{P}_{\lambda x}\})$ be an interactive quantum machine, and let $V = (p, \{V_{\lambda x u}\}, \{r^V_{\lambda x}\})$ be an interactive classical machine. Fix a security parameter $\lambda$. A \emph{nondestructive interaction} $\left (V(x) , P(x') \right )_{\rho_{\reg{AB}}}$ between $V$ and $P$ for some CQ state $\rho_{\reg{TS}}$ is an interaction in which the execution of $\left (V(x) , P(x') \right )_{\rho_{\reg{AB}}}$ is unitary (including the standard-basis measurements of the network register that take place during the execution) for all possible random inputs $u$ to $V$. More formally, for any choice of $r^V_{\lambda x}$ random strings $u_1, \dots, u_{r^V_{\lambda x}}$ used during the interaction $\left (V(x) , P(x') \right )_{\rho_{\reg{TS}}}$, there exists a unitary $U$ acting on registers (communication spaces) $\reg{N}$, $\reg{A}$ and $\reg{B}$ such that the joint state of the registers $\reg{N}$, $\reg{A}$ and $\reg{B}$ is identical after $U$ has been applied to them to their joint state after the execution of $\left (V(x) , P(x') \right )_{\rho_{\reg{AB}}}$ using the random strings $u_1, \dots, u_{r^V_{\lambda x}}$.

\end{definition}

\begin{definition}[Quantum Money Scheme, \cite{vidick2020classical}]
\label{def:qmoney}
A quantum money scheme is specified by the following objects, each of which is parametrized by a security parameter $\lambda$:
\begin{itemize}
\item A algorithm $\mathtt{Bank}$ taking a string $r$ as a parameter which initialises a database of valid \emph{money bills} in the form of a table of tuples $(\mathsf{id, public, secret}, \ket{\$}_{\mathsf{id}})$. \textsf{id} represents a unique identifier for a particular money bill; \textsf{public} and \textsf{secret} represent, respectively, public and secret information that may be necessary to run the verification procedure for the bill labeled by \textsf{id}; and $\ket{\$}_{\mathsf{id}}$ is the quantum money state associated with the identifier \textsf{id}. The string $r$ should determine a classical map $H_r$ such that $H_r(\mathsf{id}) = (\mathsf{public, secret})$.

\item A verification procedure $\texttt{Ver}(x, \textsf{public}, \textsf{secret}, \rho_W)$ that is a QPT algorithm which decides when a bill is valid.
\end{itemize}
In addition the scheme should satisfy the following conditions:
\begin{enumerate}
\item Completeness: for any valid money bill $(\mathsf{id, public, secret}, \ket{\$}_{\mathsf{id}})$ in the database created by $\mathtt{Bank}$,
\[\Pr\big( \texttt{Ver}(\reg{id}, \textsf{public}, \textsf{secret}, \proj{\$}_{\mathsf{id}})\big) \,\geq \, c_M(\lambda)\;,\]
for some function $c_M(\cdot)$.
We refer to $c_M$ as the \emph{completeness parameter} of the money scheme. 
\item No-cloning: 
Consider the following game played between a challenger and an adversary: the challenger selects a valid money bill $(\mathsf{id, public, secret}, \ket{\$}_{\mathsf{id}})$ and sends $(\mathsf{id, public,} \ket{\$}_{\mathsf{id}})$ to the adversary; the adversary produces a state $\sigma_{AB}$. Then for 
any~\footnote{Many quantum money schemes are information-theoretically secure; however, it is also possible to consider computationally secure schemes by replacing `any' with `any QPT'.} adversary in this game, the following holds  \begin{equation*}
\underset{r}{\mathrm{Pr}} \big(
\texttt{Ver}(\reg{id}, \textsf{public}, \textsf{secret}, \Tr_B(\sigma_{AB})) = 1
\end{equation*}
\begin{equation*}
 \texttt{Ver}(\reg{id}, \textsf{public}, \textsf{secret}, \Tr_A(\sigma_{AB})) = 1
\big) \,\leq\, \mu_M(\lambda)\;,     
\end{equation*}
for some function $\mu_M(\cdot)$. 
We refer to $\mu_M$ as the \emph{cloning parameter} of the money scheme. Note that the probability of the adversary's success is calculated assuming that the string $r$ which \texttt{Bank} takes is chosen uniformly at random.
\end{enumerate}
\end{definition}

\begin{definition}[Agree-and-prove scenario for Quantum Money, \cite{vidick2020classical}]
\label{def:AaPqmoney}
 An agree-and-prove scenario $(\cF_M, \cC_M, \cR_M)$ for a quantum money scenario with completeness parameter $c_M$ and cloning parameter $\mu_M$' consists of
\begin{enumerate}
\item \textbf{Setup functionality} $\cF_M(1^\lambda)$: The setup should run an initialization procedure $\mathtt{init}_M$ that instantiates a database $B_M$ whose records are of the form (and the distribution) that $\mathtt{Bank}$ would have produced running on a uniformly random input $r$. The setup should also return a specification of how the following oracles should be implemented: 
\begin{itemize}
\item $\mathcal{O}_{\cF_M}(I, \texttt{id})$: returns an identifier $\mathsf{id}$ such that the bill $(\mathsf{id, public, secret}, \ket{\$}_{\mathsf{id}})$ is in $B_M$.
\item $\mathcal{O}_{\cF_M}(\cdot, \texttt{public}, \mathsf{id})$: Returns the \textsf{public} string associated with \textsf{id}. Returns $\perp$ if no record in $B_M$ with the identifier \textsf{id} exists.
\item $\mathcal{O}_{\cF_M}(I, \texttt{getMoney}, \mathsf{id})$: If no record in $B_M$ with identifier \textsf{id} exists, returns $\perp$. Otherwise, returns $\ket{\$}_{\mathsf{id}}$ the first time it is called. If called again with the same \textsf{id} argument, returns $\perp$.
\item $\mathcal{O}_{\cF_M}(V, \texttt{secret}, \mathsf{id})$: accesses $B_M$ and returns the \textsf{secret} string associated with \textsf{id}. Returns $\perp$ if no record in $B_M$ with the identifier \textsf{id} exists.
\end{itemize}
\item \textbf{Agreement relation} $\mathcal{C}^{\mathcal{O}_{\cF_M}}(1^\lambda, \mathsf{id})$: outputs 1 if and only if a record in $B_M$ with identifier $\mathsf{id}$ exists.
\item \textbf{Proof relation} $\mathcal{R}^{\mathcal{O}_{\cF_M}}(1^\lambda, x, \rho_W)$: interprets $x$ as an \textsf{id} (outputting $\perp$ if this fails), sets $\textsf{public} \leftarrow \mathcal{O}_{\cF_M}(V, \texttt{public}, x)$ and $\textsf{secret} \leftarrow \mathcal{O}_{\cF_M}(V, \texttt{secret}, x)$, and executes $\texttt{Ver}(x, \textsf{public}, \textsf{secret}, \rho_W)$. 
\end{enumerate}
\end{definition}

\begin{theorem}[Proposition 4.3,~\cite{vidick2020classical}]
\label{thm:no-nondestructive}
Let $\lambda$ be a security parameter, let $(\cF, \cC, \cR)$ be an agree-and-prove scenario, and let $\cK = (\cI, P_1, P_2, V_1, V_2)$ be a protocol for $(\cF, \cC, \cR)$ with a classical honest verifier $V = (V_1, V_2)$, knowledge error $\kappa$ and extraction distance $\delta$. Let $\hat P = (\hat{P}_1,\hat{P}_2)$ be a prover for $\cK$.

Let $\hat{I}$ be any input generation algorithm, and $x$ and $\rho_{\reg{TS}}$ an instance and a CQ state respectively such that the agree phase of $\cK$, executed with $\hat{I}$, $V_1$ and $\hat{P}_1$, has a positive probability of ending with $x$ being agreed on, and such that the joint state of $st_V$ and $st_P$ conditioned on $x$ being agreed on is $\rho_{\reg{TS}}$.

Suppose further that (i) the interaction $\left (V_2(x), \hat P_2(x) \right )_{\rho_{\reg{TS}}}$ is nondestructive, (ii) the oracle $\cO_\cF$ does not keep state during the second phase of the protocol, i.e.\ any query to it by $V_2$ or $\hat{P}_2$ can be repeated with the same input-output behavior, and (iii) the success probability of $\hat P_2$ conditioned on instance $x$ being agreed on is at least $\kappa$. Then there exists a (possibly inefficient) procedure $A$ 
such that the following holds. 

Let $\cR^{\cO_\cF}_{\lambda x}(\cdot)$ be the function such that $\cR^{\cO_\cF}_{\lambda x}(\rho) = \cR^{\cO_\cF}(1^\lambda, x, \rho)$, and let the single-bit-valued function $\big(\cR^{\cO_\cF}_{\lambda x}\big)^{\otimes 2}(\cdot)$ be the function whose output is the $\mathrm{AND}$ of the outcomes obtained by executing the tensor product of two copies of $\cR^{\cO_\cF}_{\lambda x}(\cdot)$ on the state that is given as argument. Then the procedure $A$, given as input $x$, a copy of a communication transcript from the agree phase that led to $x$, and black-box access to $V_2$ and $\hat{P}_2$ as interactive machines (including any calls they might make to $\cO_\cF$) running on $\rho_{\reg{TS}}$, with power of initialisation for $\hat P_2$, can produce a state $\sigma$ such that
\begin{equation*}
\label{eq:A-success-prob}
\mathrm{Pr}[\big(\cR^{\cO_\cF}_{\lambda x}\big)^{\otimes 2}(\sigma) = 1] > 1 - 2\delta - \mathsf{negl}(\lambda).
\end{equation*}
\end{theorem}




\section{Remote State Preparation: Security and Construction from \texorpdfstring{\cite{cojocaru2019qfactory}}{}} \label{app_B_qfactory}

In this section we present a concrete remote state preparation protocol proposed in~\cite{cojocaru2019qfactory}. 
\medskip
\begin{breakablealgorithm}
\caption{4-states QFactory: classical delegation of $\ket{+_{\theta}}$ states ($\theta \in \{0, \pi/2, \pi, 3\pi/2\}$) (\cite{cojocaru2019qfactory})} \label{protocol:rsp}
\vspace{\baselineskip}
\textbf{Requirements:} 
A 2-regular trapdoor one-way family $\mathcal{F}$ and homomorphic hardcore predicate $\{h_k\}$.
\begin{enumerate}
    \item Preimages superposition
    \begin{enumerate}
    \item Alice runs the algorithm $(k,t_k) \leftarrow \text{Gen}_{\mathcal{F}}(1^n)$.
    \item Alice instructs Bob to prepare one register at $\otimes^n H\ket{0}$ and second register initiated at $\ket{0}^{m}$.
    \item Bob receives $k$ from Alice and applies $U_{f_k}$ using the first register as control and the second as target.
    \item Bob measures the second register in the computational basis, obtains the outcome $y$. The combined state is given by ${(\ket{x}+\ket{x'}) \otimes \ket{y}}$ with $f_k(x)=f_k(x')=y$ and $y\in \Ima f_k$. 
\end{enumerate}
\item Output preparation
\begin{enumerate}
    \item Bob applies $U_{h_k}$ on the preimage register $\ket{x}+\ket{x'}$ as control and another qubit initiated at $\ket{0}$ as target. Then, measures all the qubits, but the target in the $\{\frac{1}{\sqrt{2}}(\Ket{0} \pm \Ket{1})\}$ basis, obtaining the outcome $b = (b_1, ..., b_{n})$. 
    Bob applies on the unmeasured qubit (representing the output state) the operation $HR(-\pi/2)$.
    Now, Bob returns both $y$ and $b$ to Alice. 
    \item Alice using the trapdoor $t_k$ computes the preimages of $y$:
    \item Then compute: $\theta_2 := h_k(x) \xor h_k(x')$ and $\theta_1 := (\theta_2 \cdot \langle b, x \oplus x' \rangle) \oplus h_k(x)h_k(x')$
\end{enumerate}
\item \textbf{Outputs}: The quantum state that Bob has generated is (with overwhelming probability~\footnote{\label{note} the probability comes from the probability of $\mathcal{F}$ being a 2-regular homomorphic-hardcore family of functions}) the state $\Ket{+_{\theta}}$, where $\theta \in \{0, \pi/2, \pi, 3\pi/2\}$ state described using the two bits $(\theta_1, \theta_2)$, where $\theta_2$ is also known as the basis of the state. The output of Alice is the classical description $(\theta_1, \theta_2)$.
\end{enumerate}
\end{breakablealgorithm}
\vspace{\baselineskip}

In any run of the protocol, honest or malicious, the state that Alice believes that Bob has is the one described in Protocol \ref{protocol:rsp}. Therefore, the task that a malicious Bob wants to achieve, is to be able to guess, as good as it can, the description of the output state that Alice (based on the public communication) thinks Bob has produced. In particular, in our case, Bob needs to guess the bit $\theta_2$ (corresponding to the basis) of the (honest) output state.

\begin{definition}[4 states basis blindness]\label{def:4basisblind}
  We say that a protocol $(\pi_A, \pi_B)$ achieves \textbf{basis-blindness} with respect to an ideal list of 4 states \\
  $S = \{S_{\theta_1,\theta_2}\}_{(\theta_1,\theta_2) \in \{0,1\}^2}$ if:
  \begin{itemize}
  \item $S$ is the set of states that the protocol outputs, i.e.:
    \[\pr{\ket{\phi} = S_{B_1B_2} \in S \mid ((\theta_1,\theta_2),\ket{\phi}) \leftarrow (\pi_A \| \pi_B) } \geq 1 - \negl\]
  \item and no information is leaked about the index bit $\theta_2$ of the output state of the protocol, i.e for all QPT adversary $\cA$:
  \[\pr{ \theta_2 = \tilde{\theta_2} \mid ((\theta_1, \theta_2), \tilde{\theta_2}) \leftarrow (\pi_A \| \cA)} \leq 1/2 + \negl \]
  \end{itemize}
\end{definition}

\begin{theorem}[4-states QFactory is secure (\textup{\cite{cojocaru2019qfactory}})] \label{thm:security_qfac_2_0}
  Protocol~\ref{protocol:rsp} satisfies $4$-states basis blindness with respect to the ideal list of states: \\
  $S = \{\Ket{+}, \Ket{-}, \Ket{+_{\pi/2}}, \Ket{+_{3\pi/2}}\}$.
\end{theorem}

An 8-state QFactory protocol producing states in the set $\{k\pi/4\}_{k \in \{0, ..., 7\}}$ can be obtained from 2 runs of 4-state QFactory Protocol~\ref{protocol:rsp} given in~\cite{cojocaru2019qfactory}.
\subsection{Generic Construction}

We will denote with $g$ the injective, homomorphic, post-quantum OWF and with $h$ the hardcore predicate of $g$, which is also homomorphic with respect to the operation of $g$. In more details, we require: \\
Consider a fixed element of the domain of $g$, $x_0$ and for now a public function $h$ having the same domain as $g$ \\
Then, we define the function $f(x, c) = g(x + c \cdot x_0)$, where $c \in \{0, 1\}$. As $g$ is injective, we can see that $f$ is 2-regular (and one-way). Now, this function needs to be constructed and applied by Bob, but we don't want to reveal him the value of $x_0$, which is where the homomorphic property (for a one-time operation) steps in:
$$ g(x) * g(x_0) = g(x + x_0) \text{ for any operations ``*'' and ``+'' } $$
Then, to compute $f$: $f(x, c) = g(x) * (c \cdot g(x_0))$, therefore it is sufficient to send him $g(x_0)$ (and the description of $g$) for Bob to apply $f$. And as $g$ is one-way, then $x_0$ is also hidden from Bob. \\
Then, after applying a unitary corresponding to function $f$ and a series of measurements, Bob obtains the quantum state: $H^{B_1}X^{B_2}\ket{0} \in \{\Ket{0}, \Ket{1}, \Ket{+}, \Ket{-}\}$, a single qubit gate whose description is represented by the 2 bits $B_1$ and $B_2$. We call $B_1$ the basis bit, and $B_2$ the output bit. The target is to ensure that $B_1$ is completely hidden from Bob. \\
The formal description of the 2 bits is the following:
\begin{equation}
    \begin{split}
        & B_1 = h(x+x_0) \oplus h(x) \\
        & B_2 = (B_1 \cdot \langle b , (x \xor (x + x_0) \rangle) \xor h(x)h(x + x_0)
    \end{split}
\end{equation}
where $x$ is a randomly chosen preimage of $g$ and $b$ is a random bit-string. \\
But now, as we were saying we wanted to ensure $B_1$ is completely hidden from Bob who only received from Alice $g(x_0)$. \\
The, if we impose that $h$ is homomorphic in the sense: 
$$h(x + x_0) \oplus h(x) = h(x_0)$$,
then we have:
$$B_1 = h(x_0)$$
And if additionally, we impose that $h$ is a hardcore predicate with respect to function $g$, then Bob while he has $g(x_0)$ he knows nothing about $B_1 = h(x_0)$. \\
Moreover, as $B_1 = h(x_0)$, then Alice knows from the very beginning the value of the basis bit.

\subsection{Function Description}

The generation algorithm $Gen_{F}$ will output:
\begin{enumerate}
    \item $k$ - the public description of a 2-regular trapdoor (post-quantum) function $f_k$;
    \item $t_{k}$ - the trapdoor information corresponding to $f_k$;
    \item $hp$ - a hardcore predicate associated with $f_k$
\end{enumerate}

To construct $Gen_{F}$ we rely on a family of injective homomorphic trapdoor functions $\mathcal{G} = \{g_{k'}\}_{k'}$ and a hardcore predicate $h$ for $\mathcal{G}$.

\procedure [linenumbering]{ $\Gen_{\cF}(1^\lambda)$}{%
    (K', t_{K'}) \sample {Gen}_{\cG}(1^{\lambda}) \\%
    z_0 \sample Dom(g_{K'}) \\%
    y_0 \gets g_{K'}(z_0)\\%
    t_{k} \gets (t_{K'}, z_0)\\%
    k \gets (K', y_0)\\%
    f_k(z, c) := g_{K'}(z) + c \cdot y_0 =  g_{K'}(z + c \cdot z_0) \, \, \text{, where $c \in \{0, 1\}$} \\
    hp \gets h(z_0) \\
    \pcreturn (k, t_{k}, hp)%
}%

To construct the injective homomorphic trapdoor functions $\mathcal{G} = \{g_{k'}\}_{k'}$ we rely on the construction of \cite{MP12}.

In other words, to sample a function $f_{k}$, we first sample a matrix $K' \in \Z_q^{m \times n}$ using the construction of \cite{MP12} (that provides an injective and trapdoor function), a uniform vector $s_0 \in \Z_q^{n}$, an error $e_0 \in \Z_q^m$ according to a small Gaussian\footnote{but big enough to make sure the function is secure} and a random bit $d_0$, and we compute: 
\begin{equation}
    \begin{split}
        & z_0 = (s_0, e_0, d_0) \\
        & y_0 = K's_0 + e_0 + d_0 \times \begin{pmatrix}
                                    \frac{q}{2} & 0 & \dots & 0
                                        \end{pmatrix}^T \\
        &g_{K'}(s, e, d) = K' s + e + d \times \begin{pmatrix}
                                    \frac{q}{2} & 0 & \dots & 0
                                        \end{pmatrix}^T
    \end{split}    
\end{equation}
The function $f_{k}$ will then be defined as follow:
\begin{align}
  f_{k}(s,e,c,d) = K's + e + c \times y_0 + d  \times \begin{pmatrix}
    \frac{q}{2} & 0 & \dots & 0
  \end{pmatrix}^T
\end{align}
Note that $c$ and $d$ are bits, and the error $e$ is chosen in a bigger space\footnote{but small enough to make sure the partial functions $f(\cdot,\cdot,c,\cdot)$ are still injective} than $e_0$ to ensure that the function $f_{K,y_0}$ has two preimages with good probability. Moreover, if we define $h(s,e,c,d) = d$, it is easy to see that the hardcore property will directly come from the fact that under $\sf{LWE}$ assumption, no adversary can distinguish a $\sf{LWE}$ instance $K's_0 + e_0$ from a random vector, so it is not possible to know if we added or not a constant vector.

\section{One-Sided Simulation Based Secure Protocol for \texorpdfstring{\onetwooqfe}{}}\label{app:onesided}
Our scheme makes use of the following tools:
\begin{enumerate}
   \item Non-interactive post-quantum hiding, computationally binding commitment scheme ${\sf COM} = (Com, Dec)$;
 	\item A 2-regular Trapdoor One-Way Function $\mathcal{F} = (Gen_{\mathcal{F}}, Eval_{\mathcal{F}}, Inv_{\mathcal{F}})$.
    \item An argument of knowledge post-quantum zero-knowledge protocol $\Pi:=( \provzk,\verzk)$  for the NP-relation $Rel_f = \{(x = (k, r_f^B, com_f), w = (r_f^A, dec_f) \} $ such that $dec_f$ is the decommitment of $com_f$ and $k$ is (part of) the output of probabilistic algorithm $Gen_{\mathcal{F}}$ when run with internal random coins $r_f = r_f^A \oplus r_f^B $.  
\end{enumerate}

\begin{remark}
If we can employ a zero-knowledge proof in combination with a post-quantum hiding and statistically binding commitment we would get statistical security against malicious Alice in the one-sided simulation framework. The commitment scheme can be instantiated from any non-interactive statistically binding commitment scheme. Alternatively, we can rely on the protocol proposed by Baum et al \cite{BDLOP18} which proposes a statistical binding commitment scheme. 
Since we want to prove that our protocol is simulatable against malicious Alice, we require the post-quantum secure zero-knowledge protocol to enjoy the property of the argument of knowledge. One zero-knowledge proof of knowledge system that has all these features is the one proposed in~\cite{ananth2020concurrent} which is based on the LWE assumption.
As a result, we make an important observation that if we instantiate our \onetwooqfeprottwo{} protocol with the zero-knowledge proof of knowledge system of \cite{ananth2020concurrent} together with the statistical binding post-quantum hiding commitment scheme of \cite{BDLOP18} would imply that the protocol \onetwooqfeprottwo{} is simulation-based secure against \textit{unbounded} Alice.
\end{remark}

\begin{breakablealgorithm}
\caption{\texorpdfstring{\onetwooqfe}{} Protocol, \onetwooqfeprottwo{}, against Malicious Alice} \label{protocol:quot_1_2_malicious_alice}

\noindent \textbf{Inputs:}
\begin{enumerate}
    \item Sender (Bob): single qubit state $\Ket{\psi_{in}}$
    \item Receiver (Alice): $b \in \{0,1\}$
\end{enumerate}
\begin{enumerate}[leftmargin=*]
\item \textbf{Alice's Verification and Setup}
\begin{enumerate}
	\item[1.1] Alice samples uniformly at random $r_f^A$ from $\{0, 1\}^{\lambda}$.  
	\item[1.2] Alice runs $Com(r_f^A) \rightarrow (com_f, dec_f)$ and sends $com_f$ to Bob.
	\item[1.3] Bob samples $r_f^B$ uniformly at random from $\{0, 1\}^{\lambda}$ and sends it to Alice.
	\item[1.4] Alice now run computes $r_f = r_f^A \oplus r_f^B$. She then runs $Gen_{\mathcal{F}}$ using internal random coins $r_f$ and obtains $(k, t_k, hp)$. 
	\item[1.5] Alice now sends $k$ to Bob and runs the interactive algorithm $\provzk$ on input the statement to be proven $x = (k, r_f^B, com_f)$ and the witness $w = (r_f^A, dec_f)$.
	\item[1.6] Bob runs the interactive algorithm $\verzk$ on input the statement $x$. Let $c$ be the output of $\verzk$. If $c=0$ then Bob aborts, otherwise he waits to receive another message from Alice.
	\item[1.7] Alice assigns $\theta_2 = hp$ and encodes her input as $\phi_b := b \cdot \frac{\pi}{2}$, uniformly samples $r_A  \overset{\$}{\leftarrow} \{0, 1\}$ and computes the angle $\delta$ (See, Eq~\ref{eq:delta}) and sends $\delta$ to Bob. Bob continues to the next stage. 
\end{enumerate}
\end{enumerate}
\begin{enumerate}[resume]
\item \textbf{4-states QFactory} (Protocol \ref{protocol:rsp}) 
\item \textbf{Computation on Bob's side} (Steps 4-7 of Protocol \ref{protocol:quot_1_2_hbc_alice})  
\item \textbf{Output on Alice's side} (Step 8 of Protocol \ref{protocol:quot_1_2_hbc_alice}) 
\end{enumerate}
\textbf{Output:} Alice obtains: $s_b = M_Z [R_x\left(-\frac{\pi}{2} \cdot b \right)\ket{\psi_{in}}]$.
\vspace{0.5\baselineskip}
 \end{breakablealgorithm}
\vspace{\baselineskip}

\begin{theorem}Protocol \onetwooqfeprottwo{} securely computes \onetwooqfefunct{} with one-sided simulation.
\label{lemma:one_sided_simulation}
\end{theorem}
The correctness of the modified protocol (\autoref{thm:correctness_1_2_oqfe_semihonest}) follows from the correctness of the commitment scheme $\sf{COM}$ and the correctness of $\Pi$.

To complete the proof we need to prove the following Lemma~\ref{thm:simulation_malicious_alice} and Lemma~\ref{thm:privacy_against_bob_malicious_alice}. 

\begin{lemma}[Simulation-based (computational) security against malicious Alice] Protocol \onetwooqfeprottwo{} is simulation-based secure against malicious Alice.
\label{thm:simulation_malicious_alice}
\end{lemma}
\begin{proof}\label{proof:simulation_malicious_alice}

We need to show that for any $\qpt$ adversary $Alice^*$, there exists a $\qpt$ adversary $\mathcal{S}$ for the ideal model such that:
\begin{equation*}
\begin{split}
    \{ \ideal_{\onetwooqfefunctmath, \mathcal{S}(z), Alice}(b, \Ket{\psi_{in}}) \}_{b, \Ket{\psi_{in}}, z} \approx_q   \{ \myreal_{\onetwooqfeprottwomath, Alice^*(z), Alice}(b, \Ket{\psi_{in}}) \}_{b, \Ket{\psi_{in}}, z}
 \end{split}
\end{equation*}

In other words, to show that \onetwooqfeprottwo{} is simulation-based secure against malicious receiver $Alice^*$, we have to prove that there exists a $\qpt$ simulator $S$, that by having access only to the ideal functionality \onetwooqfefunct{}, can simulate the output of any malicious $Alice^*$ who runs one execution of \onetwooqfeprottwo{} with an honest sender Bob. The simulator $S$ having oracle access to $Alice^*$ will run as a sender Bob in the real protocol \onetwooqfeprottwo.
By studying the real protocol, we notice that the 2 messages that $Alice^*$ sends to Bob (that she could be cheating about) are $k$ (the public key of the 2-regular trapdoor function) and $\delta$ (the angle that $Alice^*$ instructs Bob to use in his computation and which should depend on her input $b$).
One of the important elements to prove the security is the Argument of Knowledge property of $\Pi$, which guarantees the existence of an extractor that, given an acceptable proof for an NP statement $x$, it extracts the witness for $x$ with overwhelming probability.
At a high level, our simulator works as follows. Upon receiving a proof $\zkp$ from $Alice^*$, the simulator runs 
the PoK extractor of $\Pi$ thus obtaining $r_f^A$,  computes $r_f = r_f^A \oplus r_f^B$ and runs the algorithm  $Gen_{\mathcal{F}}$ with internal random coins $r_f$. The output of $Gen_{\mathcal{F}}$ corresponds to $(k, t_k, hp)$, where $\theta_2 = hp$.
Now using $\delta^*$ and $\theta_2$, the simulator can compute $d := \delta^* - \theta_2 \bmod 4$ and  finally to extracts 
$b^*$ by computing $b^* = d \bmod 2$.

Before providing the formal description of our simulator, we assume, without loss of generality, that the AoK extractor of
$\Pi$ is denoted as the algorithm $E$. $E$ on input the theorem $x$ interacts (in a black-box way) with the malicious prover to extract the witness for $x$. Formally, the simulator $S$ does the following steps.
\begin{enumerate}
    \item{ Receives } $com_f$ { from } $Alice^*$ 
    \item{ Samples } $r_f^B$ { uniformly at random and sends it to }$ Alice^*$ 
    \item{ Receives } $k$  and defines the statement  $x = (k, r_f^B, com_f)$ 
    \item{ Runs } $E(x)$, { where } $x = (k, r_f^B, com_f)$ { thus obtaining } $w = (r_f^A, dec_f)$
    \item{ Verifies the decommitment phase, namely if } $Dec(com_f, r_f^A, dec_f) = 1$. 
    \item{ If not, this means that $Alice^*$ cheated during the commitment and we abort } 
    \item{ Computes } $r_f = r_f^A \oplus r_f^B$ 
    \item{ Runs } $Gen_{\mathcal{F}}$ { using the randomness $r_f$, and since} {the randomness is fixed the output of $Gen_{\mathcal{F}}$ is deterministic.} 
    \item{ Denote this output $(\bar{k}, t_k, hp)$ and assign $\theta_2 = hp$, as in the real protocol } 
    \item{ Upon receiving $\sigma^*$, computes } $d = \, \delta^* - \theta_2 \, \bmod \, 4$ 
    \item{ Computes the input of $Alice^*$ as: } $b^* = d \, \bmod \, 2$ 
    \item{ Invokes the ideal functionality } \onetwooqfefunct{} { on input $b^*$ and obtains $s_{b^*}$} 
    \item{ Runs the simulator for semi-honest Alice (see the proof of Theorem~\ref{thm:privacy_against_bob}): } $S^A(b^*, s_{b^*}$) { using the randomness } $r_f$ { to compute the last round}
\end{enumerate}

We now show this simulator $S$ is a good simulator, i.e. the output of $S$ is computationally indistinguishable from the output of $Alice^*$ in the real-world experiment. We note that there are only two scenarios in which the simulator would fail:

\begin{enumerate}
    \item $Alice^*$, to generate the public key $k$, uses a randomness $\bar{r}_f^A$ different than the one she committed to ($r_f^A$), or $k$ is not in the domain of $Gen_{\mathcal{F}}$.
    \item $Alice^*$  biases the randomness used to run $Gen_{\mathcal{F}}$ by constructing an opening for the commitment which depends on $r_f^B$.
\end{enumerate}

Loosely speaking, the simulator fails if $Alice^*$ breaks the \textit{binding} of the commitment or the argument of knowledge property of $\Pi$.

\begin{enumerate}[label=\roman*]
    \item In the first scenario, we can immediately use $Alice^*$ to construct a reduction to the AoK property of $\Pi$.
    \item In the second scenario, we want to use the malicious Alice to construct a reduction to the binding of the commitment scheme. The reduction works as follows. 
        \begin{enumerate}
    	\item  Interact with $Alice^*$ as Bob would do, and upon receiving the proof $\zkp$, run the extractor to
    obtain the opening of the commitment.  
        \item Rewind $Alice^*$ and sample a randomness ${r_f^B}'$ such that ${r_f^B}' \neq r_f^B$. Then send again ${r_f^B}'$ to $Alice^*$.
        \item As a result we receive from $Alice^*$ a new tuple $(k', {\delta^*}')$.
        \item Run again the extractor of $\Pi$ on input $(x',zkp')$ where $x' = (k', r_f^B, com_f)$ and obtain a new witness $w' = ({r_f^A}', {dec_f}'$).
        \item If ${r_f^A}' \neq r_f^A$, then this means we have found 2 decommitments $(r_f^A, dec_f)$ and $({r_f^A}', {dec_f}')$ for the same commitment $com_f$ with $r_f^A \neq {r_f^A}'$, and as a result we break the binding property of $\sf{COM}$. If ${r_f^A}' \neq r_f^A$ then restart from the beginning.
    \end{enumerate}
\end{enumerate}

\end{proof}

 \begin{lemma}
 [Privacy against Malicious Bob]\label{thm:privacy_against_bob_malicious_alice}
 The \onetwooqfe{} Protocol \onetwooqfeprottwo{} is private against malicious Bob.
 \end{lemma}
 \begin{proof}

  The proof follows directly from \autoref{thm:privacy_against_bob} and the quantum secure zero-knowledge property of $\Pi$.
  
  \end{proof}

\section{Proofs of Results \label{App_C_all_proofs}}




\subsection{Proofs from Section~\ref{sec:1_2_oqfe_semihonest}} \label{App_D_proofs_Sec_5}

\subsubsection{\texorpdfstring{\onetwooqfe{}}{} Correctness: Proof of Theorem~\ref{thm:correctness_1_2_oqfe_semihonest} \texorpdfstring{\\ \\}{}}
\label{app:correctness} 

\noindent\textbf{Theorem~\ref{thm:correctness_1_2_oqfe_semihonest}}\;\textbf{(Correctness)} \textit{In an honest run of \onetwooqfe{} Protocol~\ref{protocol:quot_1_2_hbc_alice}, when both parties follow the protocol specifications, Alice obtains the outcome $s_b = M_Z [R_x\left(-b \cdot \frac{\pi}{2} \right) \ket{\psi_{in}}]$, where $b$ is Alice's input and $\ket{\psi_{in}}$ is Bob's input.}

\begin{proof}
Using the correctness of the 4-states QFactory protocol (Protocol \ref{protocol:rsp}), Alice fixes $\theta_2 \in \{0, 1\}$ and Bob obtains at the end of the protocol the state $\Ket{\psi_A} = \Ket{+_{\theta}}$ where  $\theta$ can be described using the 2 bits $\theta_1\theta_2 \in \{0, 1\}^2$. \\
Now let us examine the computations performed in Steps 3-6 by Bob. \\

After Steps 3 and 4, he performs:
\begin{equation}
    (M_Z \otimes I)(H \otimes I) (CZ \ket{\psi_{in}} \otimes \ket{\psi_A})
\end{equation}
If the measurement outcome is $m_0$ the unmeasured qubit $\ket{\psi_1}$ becomes:
\begin{equation}
    \ket{\psi_1} = X^{m_0}R_z((-1)^{m_0}\theta) H \Ket{\psi_{in}}
\end{equation}

Bob uses now the measurement angle received from Alice:
\begin{equation}
     \delta = \phi_b + \theta_2 \cdot \frac{\pi}{2} + r_A \cdot \pi = (b + \theta_2) \cdot \frac{\pi}{2} + r_A \cdot \pi
\end{equation}
Then, Bob computes performs the following quantum measurement using the angle $\delta$ in Step 5:
\begin{equation}
    (M_Z \otimes I)(H R_z(-\delta) \otimes I_2) \ [CZ(\Ket{\psi_1} \otimes \ket{+})]
\end{equation}
If the measurement outcome is $m_1$ the unmeasured qubit $\ket{out_c'}$ is equal to:
\begin{equation}
    \begin{split}
        & \ket{out_b'} = X^{m_1} H R_z(-\delta)\Ket{\psi_1} \\
        \end{split}
\end{equation}
By replacing $\ket{\psi_1}$ we get:
\begin{equation}
    \begin{split}
        \ket{out_b'} &= X^{m_1} H R_z\left(-((b + \theta_2) \cdot \frac{\pi}{2} + r_A \cdot \pi)\right) X^{m_0}R_z\left((-1)^{m_0} \theta \right) H \Ket{\psi_{in}} \\
        &= X^{m_1}Z^{m_0}HR_z\left[-(-1)^{m_0}((b + \theta_2) \cdot \frac{\pi}{2} + r_A \cdot \pi)\right] \cdot \\
        & \cdot R_z\left[(-1)^{m_0} (\theta_1 \cdot \pi + \theta_2 \cdot \frac{\pi}{2})\right] H \Ket{\psi_{in}} \\
        &= X^{m_1}Z^{m_0}HR_z\left[(-1)^{m_0}( -(b + \theta_2) \cdot \frac{\pi}{2} - r_A \cdot \pi + \theta_1 \cdot \pi + \theta_2 \cdot \frac{\pi}{2} )\right] H \Ket{\psi_{in}} \\
        &= X^{m_1}Z^{m_0}HR_z\left[-(-1)^{m_0}b \cdot \frac{\pi}{2} + (-1)^{m_0} (\theta_1 - r_A) \cdot \pi \right] H \Ket{\psi_{in}} \\
        &= X^{m_1}Z^{m_0}HR_z\left[-(-1)^{m_0}b \cdot \frac{\pi}{2}\right] R\left[ (-1)^{m_0} (\theta_1 - r_A) \cdot \pi \right] H \Ket{\psi_{in}} \\
        \end{split}
\end{equation}
Using the relations: $R_z((-1)^{b_1}b_2\pi) = R_z(b_2 \pi) = Z^{b_2}$ and $R_z((-1)^{b_1}b_2 \frac{\pi}{2}) = Z^{b_1b_2}R_z(b_2\frac{\pi}{2})$, for any $b_1, b_2 \in \{0, 1\}$ we get:
\begin{equation}
    \begin{split}
        \ket{out_b'} &= X^{m_1}Z^{m_0} H  Z^{m_0 \, \cdot \, b} R_z\left(- b \cdot \frac{\pi}{2} \right) R_z\left[ (-1)^{m_0} (\theta_1 - r_A) \cdot \pi \right] H \Ket{\psi_{in}} \\
        &= X^{m_1 \, \oplus \, (m_0 \, \cdot \, b)} Z^{m_0} H  R_z\left[- b \cdot \frac{\pi}{2} \right] Z^{\theta_1 \oplus r_A}  H \Ket{\psi_{in}} \\
        &= X^{m_1 \, \oplus \, (m_0 \, \cdot \, b) \, \oplus \, \theta_1 \oplus \, r_A} Z^{m_0}H R_z\left(- b \cdot \frac{\pi}{2} \right) H\Ket{\psi_{in}} 
    \end{split}
\end{equation}

Then, Bob applies the final corrections in Step 5:
\begin{equation}
    \begin{split}
        & \ket{out_b} = X^{m_1} Z^{m_0} \ket{out_b'} \\
        & \ket{out_b} = X^{(m_0 \, \cdot \, b) \, \oplus \, \theta_1 \oplus \, r_A} H R_z\left(- b \cdot \frac{\pi}{2} \right) H\Ket{\psi_{in}}
    \end{split}
\end{equation}
In step 6, Bob measures this quantum state in the computational basis and sends the outcome $\bar{s_b}$ together with $m_0$ and $m_{qf}$ to Alice.
\begin{equation}
    \begin{split}
        \bar{s_b} = M_Z \ket{out_b} = M_Z X^{(m_0 \, \cdot \, b) \, \oplus \, \theta_1 \oplus \, r_A} H R_z\left(- b \cdot \frac{\pi}{2} \right) H\Ket{\psi_{in}} \\   
        \bar{s_b} = [(m_0 \, \cdot \, b) \, \oplus \, \theta_1 \oplus \, r_A] \oplus M_Z H R_z\left(- b \cdot \frac{\pi}{2} \right) H\Ket{\psi_{in}}
    \end{split}
\end{equation}
On Alice side, in Step 7, she first uses her trapdoor key $t_k$ and computes from $m_{qf}$ the value of $\theta_1$. \\
Then, we can see that Alice by computing: $s_b = \bar{s_b} \oplus (m_0 \, \cdot \, b) \, \oplus \, \theta_1 \oplus \, r_A$ she will obtain:
\begin{equation}
    s_b = M_Z H R_z\left(- b \cdot \frac{\pi}{2} \right) H\Ket{\psi_{in}} = M_Z R_x \left(- b \cdot \frac{\pi}{2} \right) \Ket{\psi_{in}}
\end{equation}

\end{proof}

\subsubsection{Security against Semi-honest Alice: Proof of Theorem~\ref{thm:simulation_hbc_alice} \texorpdfstring{\\ \\}{}}
\label{app:simulation_hbc_alice}
\noindent\textbf{Theorem~\ref{thm:simulation_hbc_alice}}\;\textbf{(Simulation-based statistical security against semi-honest Alice)} \textit{The \onetwooqfe{} Protocol, \onetwooqfeprotone{}, (Protocol~\ref{protocol:quot_1_2_hbc_alice}) securely computes \onetwooqfefunct{} in the presence of semi-honest adversary Alice.}

 \begin{proof}\label{proof:simulation_hbc_alice}
We need to prove that there exists a $\ppt$ simulator $\mathcal{S}^A$ that given Alice's input and output can simulate the view of Alice in \onetwooqfeprotone, such that the output of the simulator and the real view of Alice are indistinguishable to any \textit{unbounded distinguisher}:
    \begin{equation}
        \{\mathcal{S}^A(b, s_b)\}_{\substack{b \in \{0, 1\} \\ \Ket{\psi_{in}} \in \mathcal{H}_2 }} \approx_u \{view_{A}^{\onetwooqfeprotonemath}(b, \Ket{\psi_{in}})\}_{\substack{b \in \{0, 1\} \\ \Ket{\psi_{in}} \in \mathcal{H}_2 }}  
    \end{equation}

The view of Alice in \onetwooqfeprotone{} consists of an input $b$, the randomness $r^A$ and all the messages received from Bob during the protocol. More specifically, we have:
\begin{equation}
    {view}_{A}^{\onetwooqfeprotonemath}(b, \Ket{\psi_{in}})\} = \{b, r^A, (m_{qf}, m_0, \bar{s}_b)\} 
\end{equation}
The internal random coins $r^A$ consists of the random bit $r_A$ and the randomness used for the algorithm $Gen_{\mathcal{F}}$, denoted by $r_f^A$.

We construct the simulator $\mathcal{S}^A$ in the following way: \\

\procedure[linenumbering]{ $\mathcal{S}^A(b, s_b) $}{%
          \tilde{r}^A \sample  \mathcal{R}^A \pccomment{ $\mathcal{R}^A$ is the space of randomness of Alice, and $\tilde{r}^A$ includes $\tilde{r}_f^A$ and $\tilde{r}_A$ }\\%
          \tilde{\bar{s}}_b \sample \{0, 1\} \\
          (\tilde{k}, \tilde{t}_{k}, hp) \gets Gen_{\mathcal{F}}(1^\lambda) \\
          \tilde{y} \sample Im(f_{\tilde{k}}) \, , \,\tilde{m} \sample \{0, 1\}^n \\
          \tilde{{m}}_{qf} \gets (\tilde{y}, \tilde{m}) \\%
          \tilde{{\theta}}_1 \gets Inv_{\mathcal{F}}(\tilde{m}_{qf}, \tilde{t}_{k}) \\    
        \pcif (b == 1) \pcthen \\
          \t \tilde{\gamma} \gets \tilde{\bar{s}}_b \oplus s_b \oplus \tilde{{\theta}}_1 \oplus \tilde{r}_A \pccomment{$\tilde{\gamma} := \tilde{m}_0\cdot b$ can be computed from Step 8 of Protocol~\ref{protocol:quot_1_2_hbc_alice}. }\\%
            \t  \tilde{m}_0 \gets \tilde{\gamma} \\
          \pcelse \\
            \t \tilde{\gamma} \sample \{0, 1\} \\
            \t \tilde{m}_0 \gets \tilde{\gamma} \\ 
          \pcreturn (b, \tilde{r}^A, (\tilde{m}_{qf}, \tilde{m}_0, \tilde{\bar{s}}_b))%
}
\\

In order to prove that the output of $\mathcal{S}^A$ and the real view of Alice are indistinguishable, we need to show that the following distributions $D_1$ and $D_2$ are indistinguishable (where $D_1$ corresponds to Alice's view and $D_2$ with $\mathcal{S}^A$'s output): 

\begin{equation}
    \begin{split}
    & D_1 = \{b,  r_A,  m_{qf}, m_0, \bar{s}_b \}_{b, \Ket{\psi_{in}}} \\
    & D_2 = \{ b, \tilde{r}_A, \tilde{m}_{qf}, \tilde{m}_0, \tilde{\bar{s}}_b \}_{b, \Ket{\psi_{in}}}
\end{split}
\end{equation}
From step 7 of simulator $\mathcal{S}^A$ and Step 8 of Protocol~\ref{protocol:quot_1_2_hbc_alice}, we can also write the distribution $D_1$ and $D_2$ as:

\begin{equation}
    \begin{split}
    & D_1' = \{b, r_A,  m_{qf}, m_0, {\bar{s}}_b \}_{b, \Ket{\psi_{in}}} \\
    & D_2' = \{ b, \tilde{r}_A, \tilde{m}_{qf}, \tilde{\gamma}, \tilde{\bar{s}}_b \}_{b, \Ket{\psi_{in}}}
\end{split}
\end{equation}
Independent of the value of the bit $b$, $m_0$ and $\tilde{\gamma}$ are indistinguishable. 

Since ($r_A$, $\tilde{r}_A$) and ($\bar{s}_b$, $\tilde{\bar{s}}_b$) are sampled uniformly and independently at random from \zo, the task of the distinguisher can be equivalently seen as distinguishing these two distributions:

\begin{equation}
    \begin{split}
    & D_1'' = \{b,   m_{qf}, s_b \oplus {\theta}_1 \}_{b, \Ket{\psi_{in}}} \\
    & D_2'' = \{ b, \tilde{m}_{qf}, s_b \oplus \tilde{\theta}_1 \}_{b, \Ket{\psi_{in}}}
\end{split}
\end{equation}
Since $\theta_1$ and $\tilde{\theta}_1$ are generated from $m_{qf}$ and $\tilde{m}_{qf}$, respectively, which are sampled uniformly and independently at random, the task of distinguisher can be equivalently seen as distinguishing these two distributions:
\begin{equation}
    \begin{split}
    & D_1''' = \{b,   m_{qf}\}_{b, \Ket{\psi_{in}}} \\
    & D_2''' = \{ b, \tilde{m}_{qf} \}_{b, \Ket{\psi_{in}}}
\end{split}
\end{equation}

Finally, in the real protocol $m_{qf}$ consists of $y \in Im(f_k)$ and a bitstring $m \in \{0, 1\}^n$. They represent outcomes of Bob's measurements inside QFactory protocol (thus before $\ket{\psi_{in}}$ was even used) and irrespective of $b$, in an honest run they occur with equal probability: $ \pr{y} = \frac{1}{| \Ima f |}$, $\pr{m} = \frac{1}{2^n}$. Therefore, as in $\tilde{m_{qf}}$ we sample $\tilde{y}$ uniformly from the the image of $f_{k'}$ and $\tilde{m}$ uniformly at random from $\{0, 1\}^n$, this makes $m_{qf}$ and $\tilde{m_{qf}}$ statistically indistinguishable. 
This shows that $D_1'$ and $D_2'$ are indistinguishable against the unbounded distinguished $\mathcal{D}$, which concludes the proof.
\end{proof}

\subsubsection{Privacy against Malicious Bob: Proof of Theorem~\ref{thm:privacy_against_bob} \texorpdfstring{\\ \\}{}}
\label{app:privacy_against_bob}

\noindent\textbf{Theorem~\ref{thm:privacy_against_bob}}\;\textbf{(Privacy against Malicious Bob)} \textit{The \onetwooqfe{} Protocol~\ref{protocol:quot_1_2_hbc_alice} \onetwooqfeprotone{} is private against malicious Bob.}

\begin{proof}\label{proof:privacy_against_bob}
To show that the \onetwooqfeprotone{} is private against Bob, it suffices to show that Bob cannot distinguish between the cases when Alice has input $b = 0$ or input $b = 1$. 

In other words, we need to show that:
For any $\qpt$ $Bob^*$ (interacting with Alice in \onetwooqfeprotone) and for any auxiliary input $z$, we have:
\begin{equation}
\{view_{Bob^*}(Bob^*(z),Alice(0))\} \approx_q \{view_{Bob^*}(Bob^*(z),Alice(1))\}
    \label{eq:view_general_bob}
\end{equation}

The view of a malicious $Bob^*$ in \onetwooqfeprotone{} when he has auxiliary input $z$ and Alice has input $b$ is defined as:
\begin{equation}
    view_{Bob^*}(Bob^*(z),Alice(b)) = (z, r^B, (k, \delta))   
\end{equation}
where $r^B$ is the randomness of $Bob^*$. \\
The transcript received by $Bob^*$ from Alice during \onetwooqfeprotone{} consists of:
\begin{enumerate}
    \item $k$ - the public key obtained by Alice when running $Gen_{\mathcal{F}}$ corresponding to a 2-regular trapdoor (post-quantum) function $f_k$;
    \item $\delta$ - represents the measurement angle that Bob is instructed to use in the quantum computations he is performing in Stage 3 of the Protocol \ref{protocol:quot_1_2_hbc_alice};
\end{enumerate}
The internal random tape of $Bob^*$ contains $r_B$. We will prove by contradiction that if there exists a $\qpt$ distinguisher that can distinguish between the 2 views of $Bob^*$ in the cases Alice's input is $b = 0$ and respectively $b = 1$, then there exists a $\qpt$ algorithm that can break the \textit{4-states basis blindness} property (Definition~\ref{def:4basisblind}) of QFactory, or equivalently, the hardcore property of the basis bit $\theta_2$. More specifically, we assume that there exists a $\qpt$ algorithm $\mathcal{A}$ that on input $(r^B, k, \delta)$ can output Alice's input $b$ with probability $\frac{1}{2} + \frac{1}{p}$ and we will construct an algorithm $\mathcal{A}'$ that can break the hardcore property of the basis $\theta_2$ with probability $\frac{1}{2} + \frac{1}{p}$.  This implies that if $\mathcal{A}$ succeeds to distinguish the 2 views with inverse polynomial probability, the same applies to the hard-core predicate property, and hence we reach a contradiction.

\procedure[linenumbering]{$\mathcal{A}'(k)$}{
 r_A \sample \{0, 1\} \, , \, b \sample \{0, 1\} \, , \, B_2 \sample \{0, 1\} \\
    \tilde{\delta} \gets b + B_2 + 2r_A \bmod 4 \\
	\tilde{b} \gets \mathcal{A}(k, \tilde{\delta}) \\
	\pcif (\tilde{b} = b) \pcthen \tilde{\theta_2} \gets B_2 \\
    \pcelse  \tilde{\theta_2} \gets B_2 \oplus 1 \\
	 \pcreturn \tilde{\theta_2}  
}\\ 

Now to compute the probability that $\mathcal{A}$ break the hardcore predicate, we first consider 2 cases: i) $B_2 = \theta_2$ and ii) $B_2 \neq \theta_2$, whereas $B_2$ is sampled uniformly, each occur with probability $\frac{1}{2}$. The first case corresponds to the view of the protocol when Alice's input is $b$ and the second case corresponds to the view of the protocol when Alice's input is $1 \oplus b$. 
Therefore, we have:
\begin{equation} \nonumber
    \begin{split}
        &Pr[\mathcal{A}'(k) = \theta_2] = Pr[\mathcal{A}'(k) = \theta_2 \, | \, B_2 = \theta_2] \cdot Pr[B_2 = \theta_2] + \\
        & + Pr[\mathcal{A}'(k) = \theta_2 \, | \, B_2 = 1 \oplus \theta_2] \cdot Pr[B_2 = 1 \oplus \theta_2] \\
        &= \frac{1}{2} (Pr[\mathcal{A} \text{ outputs } b \, | \, \text{Alice input$ = b$}] 
         + Pr[\mathcal{A} \text{ outputs } 1 \oplus b \, | \, \text{Alice input$ = 1\oplus b$}]) \\
         &= \frac{1}{2} \left( \frac{1}{2} + \frac{1}{p} + \frac{1}{2} + \frac{1}{p} \right) = \frac{1}{2} + \frac{1}{p}
    \end{split}
\end{equation}
\end{proof}





\subsubsection{Security against Malicious Alice: Proof of Lemma~\ref{thm:simulation_malicious_alice_oqfe} \texorpdfstring{\\ \\}{}}
\label{app:simulation_malicious_alice_oqfe}

\noindent\textbf{Lemma~\ref{thm:simulation_malicious_alice_oqfe}}\;\textbf{(Simulation-based Security Malicious Alice)} \textit{The \oqfe{} Protocol~\ref{protocol:oqfe_protocol} is simulation-based secure against malicious Alice.}

\begin{proof}
We need to show that for any adversary $Alice^*$ there exists a $\qpt$ adversary $\mathcal{S}$ for the ideal model such that:
\begin{equation*}
\{ \ideal_{\mathcal{F}_{\sf{CQ}}, \mathcal{S}(z), Alice}(\Phi, \Phi_{in}) \}_{\Phi, \Phi_{in}, z} \approx_q \{ \myreal_{\pi_{\sf{Q2PC}}, Alice^*(z), Alice}(\Phi, \Phi_{in}) \}_{\Phi, \Phi_{in}, z}
\end{equation*}

The proof will follow closely the steps of the proof of Theorem~\ref{lemma:one_sided_simulation}, hence we only provide a sketch of the proof highlighting the main differences.

 We have to prove that there exists a $\qpt$ simulator $S$, that by having access only to the ideal functionality $\mathcal{F}_{\sf{CQ}}$, can simulate the output of any malicious $Alice^*$ who runs one execution of $\pi_{\sf{Q2PC}}$ with an honest sender Bob. The simulator $S$ having oracle access to $Alice^*$ will run as a sender Bob in the real protocol.
 The simulator runs the argument of knowledge extractor for $\Pi^\star$ (which exists by definition), and 
 extracts the input of the adversary (i.e., it extracts $\phi_{i,j}$ for all $i\in[n],j\in[m]$). 
 The simulator now can invoke the ideal functionality $\mathcal{F}_{\sf{CQ}}$ to
 obtain the output.
 Then the simulator extracts the trapdoors for all the trapdoor $\sf{OWFs}$ keys following the same procedure
 of the simulator of Theorem~\ref{lemma:one_sided_simulation}, to make sure that $Alice^\star$ was behaving honestly.
 
 From this point on we are guaranteed that $Alice^\star$ behaves honestly by the soundness of $\Pi'$. Moreover, the soundness of $\Pi'$ guarantees that the input that $Alice^\star$ is using to compute the
 values $\{\delta_{i,j}\}_{i\in[n],j\in [m]}$ is compatible with the input  $\{\phi_{i,j}\}_{i\in[n],j\in[m]}$ extracted by the simulator.
 Hence, the simulator can act as the semi-honest simulator for the \ubqc{} stage of the protocol.
\end{proof}




\subsection{Proofs of Section~\ref{sec:full_simulation_two_pc}}  \label{App_F_proofs_Sec_7}

\subsubsection{Proof of Theorem~\ref{thm:full_sim_qtwopc} \texorpdfstring{\\ \\}{}} \label{app:theorem_full_sim}

\noindent\textbf{Theorem~\ref{thm:full_sim_qtwopc}}\; \textit{Protocol~\ref{protocol:full_sim_qtwopc} is a secure black-box \qtwopc{} (as defined in Def.~\ref{def:Qtwopc}), assuming that Bob has a classical description of his input.}

\begin{proof}
Firstly, to show that Protocol~\ref{protocol:full_sim_qtwopc} is simulation-based secure against malicious $Bob^*$, we have to prove that there exists a QPT simulator $S$, that by having access only to the ideal functionality of \qtwopc{} , can simulate the output of any malicious $Bob^*$ who runs one execution of Protocol~\ref{protocol:full_sim_qtwopc} with an honest Alice. The simulator $S$ having oracle access to $Bob^*$ will run as Alice in the real protocol. \\

\noindent We construct simulator $S$ as follows:
\begin{enumerate}
    \item Receives $com_y$ from $Bob^*$;
    \item Defines statement $x = com_y$;
    \item Runs extractor $E$ for the proof of knowledge $(P_C^{(1)}, V_C^{(1)})$ for $Rel_C$, thus obtaining $w := (dec_y, y_{\ket{\psi}})$;
    \item Verifies the decommitment phase, namely if $Dec(com_y, y_{\ket{\psi}}, dec_y) = 1$;
    \item If the extraction fails then $S$ aborts;
    \item $S$ will run as Alice using the input $0$ during the $\sf{OS-Q2PC}$ protocol; 
    \item Invokes the ideal functionality for \qtwopc{} on inputs $(x, \ket{\psi_{in}})$, where $\ket{\psi_{in}}$ is the quantum state having the classical description $y_{\ket{\psi}}$. $S$ obtains the outcome of the ideal functionality $out$;
    \item Runs the verifier $V_C^{(2)}$ for the zero-knowledge $(P_Q^{(2)}, V_C^{(2)})$. If $V_C^{(2)}$ rejects then we abort;
\end{enumerate}

We now need to show the correctness of our simulator $S$, i.e. the output of $S$ is computationally indistinguishable from the output of $Bob^*$ in the real-world experiment. We notice that the simulator $S$ fails if $Bob^*$ breaks the binding of the commitment, the proof of knowledge property of $(P_C^{(1)}, V_C^{(1)})$ or the soundness property of $(P_Q^{(2)}, V_C^{(2)})$.

Finally, to obtain the full-simulation secure black-box \qtwopc{} protocol we just need to plug our one-sided simulation (black-box) \qtwopc{} Protocol~\ref{protocol:oqfe_protocol} which achieves simulation-based security against Alice and privacy against quantum Bob, in our compiler $C$ (Protocol~\ref{protocol:full_sim_qtwopc}) in order to boost to simulation security against Bob. \\
The simulation-based security against a $QPT$ Alice is preserved, due to the sequential composition of the two post-quantum zero-knowledge protocols $(P_C^{(1)}, V_C^{(1)})$ and $(P_Q^{(2)}, V_C^{(2)})$.
\end{proof}




\subsection{Proofs of Section~\ref{sec:compiler_zkpoqk}}  \label{App_G_proofs_Sec_8}

To prove Theorem~\ref{thm:zkpoqk} we need to show the following three lemmas.


\begin{lemma}
Assuming $(P_Q^{(1)}, V_C^{(1)})$ is a proof of knowledge with knowledge error $\kappa_1$ and $(P_C^{(2)}, V_C^{(2)})$ is a proof of knowledge with knowledge error $\kappa_2$, then $(P_Q, V_C)$ is a classical post-quantum proof of quantum knowledge system with knowledge error $\kappa_1 \kappa_2$.
\end{lemma}


\begin{proof}
We will construct our extractor $E$ by running sequentially two extractors: $E_1$ (for $(P_C^{(2)}, V_C^{(2)}$) and $E_2$ (for $(P_Q^{(1)}, V_C^{(1)}$).  
Consider a malicious prover $P_Q^*$ for our \zkpoqk{} Protocol~\ref{protocol:zkcpoqk}. \\
Then, we need to construct an extractor $E$ which needs to output the witness $w$ by having black-box access to $P_Q^*$. We will now construct two provers ${P_C^{(2)}}^*$ and ${P_Q^{(1)}}^*$ which will use $P_Q^*$ as a black-box.  Firstly, we will define the left interface as corresponding to the prover $P_Q^*$ and the right interface corresponding to the verifier. \\

\noindent We construct ${P_C^{(2)}}^*$ as follows:
\begin{enumerate}
    \item For each round of $(P_C^{(2)}, V_C^{(2)})$ (corresponding to only the first proof of knowledge):
    \begin{enumerate}
        \item Calls $P_Q^*$ and receives from $P_Q^*$ (left interface) a message $m_i$;
        \item Forwards $m_i$ to the right interface;
        \item Upon receiving a message $n_i$ from the right interface, ${P_C^{(2)}}^*$ sends $n_i$ to $P_Q^*$.
    \end{enumerate}
\end{enumerate}
Given that $(P_C^{(2)}, V_C^{(2)})$ is a proof of knowledge system, we know there must exist an extractor $E_1$ such that when given black-box access to our constructed ${P_C^{(2)}}^*$, $E_1$ will be able to output the witness $w_1 := sk$ as well as a state $\rho_1'$. 
We will use $\rho_1'$ as an input for the second extractor $E_2$, while we will use the witness $w_1$, when constructing our second prover ${P_Q^{(1)}}^*$. \\

\noindent We construct ${P_Q^{(1)}}^*$ as follows: 
\begin{enumerate}
    \item For each round of $(P_Q^{(1)}, V_C^{(1)})$ (corresponding to only the second proof of knowledge): 
        \begin{enumerate}
            \item Calls $P_Q^*$ on input $\rho_1'$ and receives from $P_Q^*$ a message $m_i'$;
            \item Computes $Dec(m_i', sk)$ and forwards this decrypted message to the right interface;
            \item Upon receiving a message $n_i'$ from the right interface,  ${P_Q^{(1)}}^*$ sends $n_i'$ to $P_Q^*$.
        \end{enumerate}
\end{enumerate}

Now, given that $({P_Q^{(1)}}, V_C^{(1)})$ is a proof of quantum knowledge system, there must exist an extractor $E_2$ which on input $\rho_1'$, when given oracle access to the newly constructed prover ${P_Q^{(1)}}^*$, $E_2$ will be able to output the witness which represents exactly the witness $w$ of our initial relation $Rel_Q$. 

Now, we can argue that $E_2$ will work correctly. If instead, we assume that $E_2$ does not work correctly, then we have a reduction to breaking the soundness of the final zero-knowledge (in step 6), which ensures that the messages $m_i'$ received from $P_Q^*$ are encrypted under the secret key $w_1 = sk$.

Our main extractor $E$ will just run sequentially the two extractors: $E_1$, followed by $E_2$. \\
Then, we will use the result of \cite{unruh2012quantum} (Theorem 3)
which will imply that $(P_Q, V_C)$ has knowledge error $\kappa_1 \cdot \kappa_2$ and our extractor $E$ will succeed with probability at least:
$$ \frac{1}{p(\eta)} (Pr[\langle P_Q^*(x, \rho), V_C(x) \rangle = 1] - \kappa_1(\eta) \kappa_2(\eta))^d $$
where $p$ is a polynomial function, $d$ is a constant and $\eta$ is the security parameter.
\end{proof}


\begin{lemma}
If $(P_C^{(2)}, V_C^{(2)})$ and $(P_Q^{(1)}, V_C^{(1)})$ are simulatable then $(P_Q, V_C)$ is simulatable.
\end{lemma}


\begin{proof}
As explained above, the extractor $E$ for the proof of knowledge $(P_Q, V_C)$ will first run the extractor $E_1$ for the proof of knowledge $(P_C^{(2)}, V_C^{(2)})$ and then the extractor $E_2$ for the proof of knowledge $(P_Q^{(1)}, V_C^{(1)})$. \\  
To show that $(P_Q, V_C)$ has the simulatability property we will use a hybrid argument. The first hybrid $Hyb_1$ corresponds to the real world. In the second hybrid $Hyb_2$, we replace the honest verifier $V_C^{(2)}$ with the extractor $E_1$. The indistinguishability of these two hybrids follows from a direct reduction to breaking the simulatability property of $(P_C^{(2)}, V_C^{(2)})$. Then, we introduce a third hybrid $Hyb_3$, where we replace the honest verifier $V_C^{(1)}$ with the extractor $E_2$. Similarly, the indistinguishability between $Hyb_2$ and $Hyb_3$ follows from the simulatability property of $(P_Q^{(1)}, V_C^{(1)})$. Then consider that for $(P_C^{(2)}, V_C^{(2)})$, we have that: $|Pr[D(\rho_1) = 1] - Pr[D(\rho_1')]| \leq negl_1$, for any QPT distinguishers $D$, where $\tilde{\rho_1}$ and $ \rho_1'$ represent the outputs of $P_C^{(2)}$ and $E_1$ respectively, when $V_C^{(2)}$ accepts. Analogously, for $(P_Q^{(1)}, V_C^{(1)})$, we would have: $|Pr[D(\rho_2) = 1] - Pr[D(\rho_2')]| \leq negl_2$. Hence, for $(P_Q, V_C)$ we have that the probability of distinguishing between $\tilde{\rho}$ and $ \rho'$ is at most $negl_1 + negl_2 = negl$ (where $\tilde{\rho}$ and $ \rho'$ represent the outputs of $P_Q$ and $E$ respectively, when $V_C$ accepts).
\end{proof}


\begin{lemma}
$(P_Q, V_C)$ is a post-quantum zero-knowledge proof system.
\end{lemma}


\begin{proof}
We will construct our extractor simulator $S$ by running sequentially two simulators: $S_1$ (for $(P_C^{(2)}, V_C^{(2)}$) and $E_2$ (for $(P_C^{(3)}, V_C^{(3)}$).  
Consider a malicious verifier $V_C^*$ for our \zkpoqk{} Protocol~\ref{protocol:zkcpoqk}. \\
Then, we need to construct a simulator $S$ that by having black-box access to $V_C^*$, needs to output a transcript that is (computationally) indistinguishable from the view of $V_C^*$ when interacting with prover $P_Q(x, w)$ for any $(x, w) \in Rel_Q$.

We will now construct two verifiers ${V_C^{(2)}}^*$ and ${V_C^{(3)}}^*$ which will use $V_C^*$ as a black-box. \\
We will define the left interface as corresponding to the prover and the right interface corresponding to $V_C^*$;

First, $S$ will compute $Com(0^l) \rightarrow (com_{0^l}, dec_{0^l})$, where $l$ is the size of $sk$, and will send $com_{0^l}$ to $V_C^*$, which cannot be distinguished from the real $com_{sk}$ due to the post-quantum hiding property of $Com$.

We construct ${V_C^{(2)}}^*$ as follows:
\begin{enumerate}
    \item For each round of $(P_C^{(2)}, V_C^{(2)})$ (corresponding to only the first zero-knowledge):
    \begin{enumerate}
        \item Calls $V_C^*$ and receives from $V_C^*$ (right interface) a message $m_i$;
        \item Forwards $m_i$ to the left interface;
        \item Upon receiving a message $n_i$ from the left interface, ${V_C^{(2)}}^*$ sends $n_i$ to $V_C^*$.
    \end{enumerate}
\end{enumerate}

Consequently, as $(P_C^{(2)}, V_C^{(2)})$ is a zero-knowledge system, there must exist a simulator $S_1$ such that when given black-box access to our constructed ${V_C^{(2)}}^*$, $S_1$ will be able to output a transcript that is indistinguishable from $\{View_{V_C^{(2)}} \langle P_1(x_1, w_1), {V_C^{(2)}}^*(x_1, \rho_1) \rangle\}$.

And $S$ will run this simulator $S_1$.

Next, corresponding to step 5 of Protocol~\ref{protocol:zkcpoqk}, $S$ will run $Gen(1^\lambda)$ and obtain $sk'$. For $i \in [N]$, $S$ will send $\tilde{enc}_i := Enc(sk', 0^{l_i})$ (where $l_i$ is the length of $m_i$) to $V_C^*$, which cannot be distinguished form the real $enc_i$, due to the post-quantum property of $Enc$.

Now we construct  ${V_C^{(3)}}^*$ as follows:
\begin{enumerate}
    \item For each round of $(P_C^{(3)}, V_C^{(3)})$ (corresponding to only the last zero-knowledge):
    \begin{enumerate}
        \item Calls $V_C^*$ and receives from $V_C^*$ (right interface) a message $m_i'$;
        \item Forwards $m_i'$ to the left interface;
        \item Upon receiving a message $n_i'$ from the left interface, ${V_C^{(2)}}^*$ sends $n_i'$ to $V_C^*$.
    \end{enumerate}
\end{enumerate}

Consequently, as $(P_C^{(3)}, V_C^{(3)})$ is a zero-knowledge system, there must exist a simulator $S_2$ such that when given black-box access to our constructed ${V_C^{(3)}}^*$, $S_2$ will be able to output a transcript that is indistinguishable from $\{View_{V_C^{(3)}} \langle P_3(x_2, w_2), {V_C^{(3)}}^*(x_2, \rho_2) \rangle\}$.
Finally, $S$ will run simulator $S_2$.
    
To show the correctness of our simulator $S$, we will use a hybrid argument. The first hybrid $Hyb_1$ corresponds to the real world. In the second hybrid $Hyb_2$ we replace $com_{sk}$ with $com_{0^l}$. The indistinguishability of these two hybrids follows from a reduction to breaking the post-quantum hiding property of $Com$. In the third hybrid $Hyb_3$, we replace the interaction during the first zero-knowledge proof $(P_C^{(2)}, V_C^{(2)})$ with the output of the simulator $S_1$ and the indistinguishability between the second and third hybrid holds by a reduction to the zero-knowledge property of $(P_C^{(2)}, V_C^{(2)})$. In the fourth hybrid $Hyb_4$, we replace the encrypted messages $enc_i$ with $\tilde{enc}_i$ and the indistinguishability of hybrids 3 and 4 is due to the post-quantum property of $Enc$. Finally, we have the hybrid $Hyb_5$ in which we replace the interaction during the second zero-knowledge proof $(P_C^{(3)}, V_C^{(3)})$ with the output of the simulator $S_2$ and the last two hybrids are indistinguishable due to the zero-knowledge property of $(P_C^{(3)}, V_C^{(3)})$.
\end{proof}

\end{document}